\documentclass{article}
\usepackage[utf8]{inputenc}
\usepackage{float}
\usepackage{graphicx}
\usepackage{caption}
\usepackage{bbm}
\usepackage{subcaption}
\usepackage{dcolumn}
\usepackage{bm}
\usepackage{tabularx}
\usepackage{amsmath}
\usepackage{amsfonts}

\usepackage[T1]{fontenc}
\usepackage[fontsize=12pt]{fontsize}
\usepackage{amsthm}
\usepackage{babel}
\theoremstyle{plain}
\newtheorem{thm}{Theorem}[section]
\newtheorem{proposition}[thm]{Proposition}
\newtheorem{lemma}[thm]{Lemma}
\newtheorem{definition}[thm]{Definition}
\usepackage{authblk}

\begin{document}

\title{MSc thesis- A quantum pseudo-integrable Hamiltonian impact system}
\author{Omer Yaniv}
\affil{Department of Computer Science and Applied Mathematics,Weizmann Institute of Science, Rehovot 7610001, Israel }

\maketitle

\section*{Abstract}
 A quantization of a toy model of a pseudointegrable Hamiltonian impact system is introduced, including EBK quantization conditions, a verification of Weyl's law,  the study of their wavefunctions and a study of their energy levels properties. It is demonstrated that the energy levels statistics are similar to those of pseudointegrable billiards. Yet, here, the density of wavefunctions which concentrate on projections of classical level sets to the configuration space does not disappear at large energies, suggesting that there is no equidistribution in the configuration space in the large energy limit; this is shown analytically for some limit symmetric cases and is demonstrated numerically for some nonsymmetric cases.
\section{Introduction}

Hamiltonian systems are a mathematical branch of the theory of dynamical systems, which investigates the time-evolution of a system's state,  attempting to provide information on the evolution of all possible initial states of the system. When such a system evolves continuously, it is called a continuous dynamical system, and in the smooth case it is described by a set of differential equations. \newline

\subsection{Classical Hamiltonian dynamical systems}

Hamiltonian systems arise in various fields of Physics: for example, they describe the dynamics of a body moving under conservative forces. 

A Hamiltonian system on a symplectic manifold \cite{arnol2013mathematical} \cite{meiss2007differential} $(M, \omega) $, where $\omega$ is a closed non-degenerate 2 form, e.g $\omega=\sum_{i=1}^n dq_i\wedge dp_i$ for set of canonical coordinates $(p,q)$, is defined by a smooth Hamiltonian function $H \in C^{\infty}(M,\mathbb{R})$, and $X_{H}$ which is the associated vector field to H:\newline
\begin{equation} dH(Y)=\omega(X_{H},Y)\end{equation}\newline
The dynamics of any observable $f$ obeys the equation of motion:
\begin{equation}\frac {d}{dt}f=\{f,H\}
\label{poissonbracket}
\end{equation} where  $\{f,H\}$ are called Poisson brackets and are defined as: $ \omega(X_{f},X_{H})$. In canonical coordinates $(q,p)$, we can write $X_{H}=\sum_{i=1}^n(\frac{\partial H}{\partial p_{i}}\frac{\partial}{\partial q_{i}}-\frac{\partial H}{\partial q_{i}}\frac{\partial}{\partial p_{i}})$. \newline
The Hamiltonian flow $F^{t}: M \rightarrow M $ is generated by $X_{H}$ as \begin{equation}\frac {d}{dt}F^{t}_{t=0}=X_{H}\end{equation}
$F^{t}: M \rightarrow M $ is symplectomorphism: 

\begin{equation} F^{t*}\omega =\omega 
\end{equation}

Thus, 
\begin{equation}\frac{ F^{t*}\omega^{n}}{n!} =\frac{\omega^{n}}{n!}
\label{preserving}
\end{equation}
\newline
where $\omega^{n}$ is the phase space volume differential form \newline 
e.g, given a region $R \in M$

\begin{equation}
\int_{R} \frac{\omega^{n}}{n!} = \int_{R} dq_1\wedge...\wedge dq_n \wedge dp_1\wedge...\wedge dp_n =Vol(R)
\end{equation}
Thus, eq. \ref{preserving}  implies that flow $F^{t}$ is volume preserving. \newline

We concentrate here on Hamiltonian systems of the mechanical form, where the Hamiltonian function is the sum of a diagonal kinetic energy term and a potential energy defined on $\mathbb{R}^{n}$, and the natural canonical variables are the momenta, $p_i$ and the configuration coordinates $q_i, i=1,...n$:
\[H(q,p)=\frac{1}{2}\sum_{i=1}^n p_i^2+V(q), \qquad  (q,p)\in \mathbb{R}^{n} \times \mathbb{R}^{n}\]
These are called  $n$ degrees of freedom mechanical Hamiltonian systems. The goal of the analysis is to determine the evolution of all initial conditions in the phase space $(q,p)$.
 In these coordinates the equations of motion are: 
\begin{align}\frac{d}{dt}p_{i}&=-\frac {\partial H}{\partial q_{i}}  \\ \frac{d}{dt} q_{i}&=\frac {\partial H}{\partial p_{i}}\nonumber\end{align}

The simplest type of $n$ degrees of freedom Hamiltonian systems are integrable systems. Such systems exhibit at least $n$ independent smooth constants of motions which are pairwise in involution  (their pair-wise Poisson bracket vanishes). A full description of the dynamics in such systems is provided by the Arnold- Liouville theorem \cite{arnol2013mathematical}, which tells us that for almost all values of $(q,p)$ there exists a transformation  $(q_{i},p_{i}) \rightarrow (\theta_{i}, I_{i})$ such that $H(q,p)=H(I)$ and the motion occurs along connected components of level sets of $I$. Each such component is an $n$ dimensional torus: 

\begin{align} {\dot {I}}_{i}=-{\frac {\partial H}{\partial \theta_{i}}}=0\\{\dot {\theta}_{i}}=-{\frac {\partial H}{\partial I_{i}} }=\omega_{i}\nonumber\end{align} 

Moreover, in open neighborhoods at which such transformation exists, it is smooth. 

On the other hand, Hamiltonian systems can also be chaotic, which is an extreme case of a disordered motion. The definition of chaotic systems is community dependent\cite{devaney2018introduction} 
\cite{sinai1970dynamical}\cite{sinai1989dynamical} \cite{ruelle1971nature}. \newline

In ergodic theory \cite{sarig2009lecture}, chaos is defined by the global properties of a measure preserving map on a probability space $(X,\mathcal{B},\mu,T)$ , which in our case is the symplectomorphism associated with the Hamiltonian flow. \newline 
In order to define chaos in such systems we first need to define \textbf{ergodicity}: T is said to be \textbf{ergodic} if the only T-invariant sets are of measure 0 or of measure 1. \newline
An ergodic map T is said to have the strong mixing property if for all $A,B \in \mathcal{B}$, $\lim_{k \rightarrow \infty} \mu(A \cap T^{-k}(B))=\mu(A)\mu(B)$. \newline

An ergodic map which is strongly mixing is usually said to be chaotic  within the notion of ergodic theory (yet, this definition depends on the measure $\mu$, and, notably, choosing a "natural invariant  measure" is a delicate issue in ergodic theory. For Hamiltonian systems the natural measure is the Liouville measure: Lebesgue measure restricted to the energy surface \cite{sarig2009lecture}). 

A remarkable property of ergodic systems, which is highly important in physics, relates to the properties of an observable f, defined on the system $( X,\mathcal{B},\mu,T)$ : \newline

if $T$ is ergodic and  $f\circ T \equiv f $, then f is constant almost everywhere. \newline

From this we conclude the general property that  for any measurable function $g:X \rightarrow \mathbb{R}$ the time average equals the space average:
\begin{align}
 \int_{X} g d\mu= \frac{1}{T}\lim_{T \rightarrow \infty}\int_{0}^{\infty}gdt
\end{align}

The mathematical study of ergodicity in Hamiltonian systems focuses mainly on the dynamics of billiards: systems of free particles which reflect elastically from the billiard table $D$. Formally we take $V(q)\equiv0$ in the interior of the billiard table and on its boundary  the motion is determined by elastic reflections:  $(q,p_{\perp}) \rightarrow (q,-p_{\perp}), (q,p_{\parallel}) \rightarrow (q,p_{\parallel})$ \cite{chernov2006chaotic}.\newline

In this work, we will look at Hamiltonian impact systems (HIS), which is an extension of the class of billiards: the dynamics is determined by a smooth Hamiltonian in the table interior and by elastic reflections: $(q,p_{\perp}) \rightarrow (q,-p_{\perp}), (q,p_{\parallel}) \rightarrow (q,p_{\parallel})$ on the table's boundry.

\subsection{Pseudo integrable billiards}
Pseudo integrable billiards arise in the study of plane polygonal rational billiards (polygonal tables with all corners being rational fractions of $\pi$).

The path of a particle in a billiard is independent of the energy. For polygonal billiards, it can be considered as lying in a 3 dimensional space whose coordinates are the position $q=(q_1,q_2)$ and $\theta$, the path direction. This path lies in a sequence
of replicas of the enclosure situated at
different values of $\theta$. In the case of plane polygonal rational billiards this sequence is finite(the number of possible directions in each trajectory). Since reflection at an edge is equivalent for continuing the path into a reflection of the enclosure and the sequence of directions is finite, the trajectory lies on a 2-dimensional compact surface.
Such surfaces are  two-dimensional surfaces of genus $g \geq 1$ \cite{richens1981pseudointegrable,gutkin1996geometry}.
Pseudointegrable dynamics, correspond to systems with intermediate complexity: they are not ergodic nor quasi-periodic, there can be level sets that include several ergodic components as well as bands of periodic orbits.

\subsection{Quantum dynamical systems and correspondence principle}
In the early 20th century, it was clear that physical properties of some systems cannot be predicted by classical physics. A new theory, quantum mechanics, was found to be useful in their prediction. According to quantum mechanics the physical properties of an n-degrees of freedom system are not defined by smooth function on $\mathbb{R}^{2n}$, but rather is defined by a Hermitian operator in $Op(\mathcal{S}(\mathbb{R}^{n}))$ where: $ S(\mathbb{R}^{n})=\{ \phi \in C^{\infty} | \sup_{\mathbb{R}^{n}} |x^{\alpha}\partial^{\beta} \phi|<\infty $\hfill for   all multiindices \hfill$ \alpha, \beta \}$.

Given an operator $\hat F=Q(F(q_{i},p_{i}))$ with $F\in C^{\infty}(\mathbb{R}^{2n})$, its time evolution, the operator $\hat F$  at time $t>0$ is: $\hat{F}_{t}=e^{i\frac{t\hat{H}}{\hbar}}\hat{F}e^{-i\frac{t\hat{H}}{\hbar}}$, where $\hat{H}=Q(H(q_{i},p_{i}))$ is the operator associated with the Hamiltonian. Then,  the time evolution equation, in analogy to the classical equation, becomes: \begin{align} \dot{\hat{{F}_{t}}}=-i\frac{[\hat{H},\hat{F}_{t}]}{\hbar} 
\label{hizenberg}
\end{align}
Another equivalent form to express the time evolution of the system is via the time-dependent Schrodinger equation, that considers evolving in time the wave function $\Psi(q,t)$ rather than the operator:\begin{align}
i \hbar \frac{\partial}{\partial t}\Psi(q,t) = \hat H \Psi(q,t) \end{align}
The time-independent Schrodinger equation is an eigenvalue problem, whose solutions are called the spectrum of the quantum system: \begin{align}
\hat H \Psi(q)=E\Psi(q)\end{align}

It is believed, that the classical dynamics can be derived as a limit of the quantum dynamics in the limit of high energies, $\lim \hbar\rightarrow 0$. This assumption is called the correspondence principle. \newline

In order to relate quantum and classical mechanics we are interested in the class of transformations $Q:C^{\infty}(\mathbb{R}^{2n}) \rightarrow Op(\mathcal{S}(\mathbb{R}^{n}))$, which are called quantization of classical phase space \cite{ali2005quantization}. \newline
We require Q to satisfy to following conditions: \newline
\begin{itemize}
\item  Q is linear \newline
\item  Q(1) = I, where 1 is the constant function 1, and I the
identity operator \newline
\item  for any function $\Phi : R \rightarrow R $ for which $Q(\Phi\circ f)$ and $\Phi(Q(f))$ are well-defined,
$Q(\Phi\circ f) = \Phi(Q(f) )$ \newline
\item  The operators $Q(p_i)$ and $Q(q_i)$ corresponding to the coordinate functions $F(p,q)=p_i$ , $G(p,q)= q_i$
for (i = 1,...,n) are given by
$Q(G)\Psi = q_i\Psi$ and $Q(F)\Psi = -i\hbar\frac{\partial{\Psi}}{\partial{q_i}}$ 
\end{itemize}

To obtain a correspondence between classical(as described in eq. \ref{poissonbracket}) and quantum dynamics(as described in eq. \ref{hizenberg}) we would want to find Q such that: \newline
$Q(\{f,g\})= \frac{[Q(f),Q(g)]}{i\hbar}$ \newline
However, a quantization Q that satisfies this condition for all f, g $\in C^{\infty}(\mathbb{R}^{2n})$  altogether with conditions 1 to 4 , mentioned above, doesn't exist\cite{ali2005quantization}. \newline

The common quantization scheme which is used is Weyl's quatization\cite{evans2007lectures}: \newline
For a function $a(q,p) \in C^{\infty}(\mathbb{R}^{2n})$ in classical phase space, $a^{w}$ the operator accepted by Weyl quantization is defined by:
\begin{align}
a^{w}\Psi(x)= \frac{1}{\hbar^{n}}\int_{\mathbb{R^{n}}}\int_{\mathbb{R^{n}}} e^{\frac{i<x-y,p>}{\hbar}}a(\frac{x+y}{2},p) \Psi(y) dydp \end{align}

Weyl's quantization satisfies the following condition:

\begin{align}
Q(\{f,g\})= \frac{[Q(f),Q(g)]}{i\hbar}+O(\hbar^{2})
\end{align}

Thus, in the limit $\hbar \rightarrow 0$ (so called semiclassical limit) we get a correspondence between classical and quantum dynamics, described by \textbf{Egorov's Theorem} \cite{evans2007lectures}: \begin{align}
\lVert \hat{F}_t-Q(F\circ\Phi_{t}(q,p)) \rVert=O(\hbar)\text{ where } Q(F)=\hat{F}_{0} 
\end{align}
Another important theorem regarding classic-quantum correspondence is \textbf{Weyl's theorem}\cite{evans2007lectures}:
For each $a<b$, we can define the operator $\mathbbm{1}$: 
    
\begin{align}
\mathbbm{1}(\Psi_{E_{j}})= 
\begin{cases} \Psi_{E_{j}} \phantom{-} \text{  if} : a<{E_{j}}<b \\
\phantom{-}0 \phantom{-} \phantom{-} \text{else}
 \end{cases} 
\end{align}
Weyl's theorem states that: 
\begin{align}
 N[E_{j}: a \leq E_{j} \leq b]=\sum_{j}<\mathbbm{1}\Psi_{E_{j}},\Psi_{E_{j}}>=
\frac{1}{\hbar^n}Vol(a\leq H\leq b)+o(1)
\text{ as } \hbar\rightarrow 0 
\end{align}
where $E_{j}$ and $\Psi_{E_{j}}$  are the eigenvalues and eigenvectors (correspondingly) of $\hat{H}=Q(H(q_{i},p_{i}))$ and $H(q_{i},p_{i})$ is the classical Hamiltonian.\newline

Semmiclassical approximations can also be used within the framework of Schrodinger equation. An important approximation for one degree of freedom systems is  \textbf{WKB aproximation} \cite{keller1958corrected} which is an approximation for time independent Schrodinger equation: \newline
In the limit $\hbar \rightarrow 0$, $\Psi(q)$, an eigenfunction of $\hat{H}=Q(H(q_{i},p_{i}))$ can be approximated by:
\begin{align}
\displaystyle \Psi (x)\approx C_{0}{\frac {e^{\theta +i\int \hbar ^{-1}{\sqrt {2m\left(E-V(x)\right)}}\,dx}}{\hbar ^{-1/2}{\sqrt[{4}]{2m\left(E-V(x)\right)}}}}
\end{align}

From this approximation we can derive  \textbf{EBK quantization conditions}, that are used to find the spectrum of integrable systems:

\begin{align}
I(E)=\hbar\left(n+\frac{\mu}{4}+\frac{b}{2}\right)
\end{align}

where I is the classical action of a periodic orbit, $\mu$ is the number of classical turning points along a period and $b$ is the number of impacts from the table's boundries along a period \cite{brack2018semiclassical}(chapter 2).

\subsection{Quantum chaos and ergodicity}

Quantum chaos studies how classical dynamics (integrable and non-integrable) are reflected in the properties (e.g. eigenvalues and eigenfunctions) of the correspondent quantum system.
It is accepted that in integrable systems, the distribution of the level spacing is provided by the Poisson distribution $e^{-s}$ \cite{berry1977level}, while that in chaotic systems (hereafter, meaning mixing system on energy surfaces, studied by simulating chaotic billiards) they distribute as eigenvalues of random matrix ensembles (GOE)  \cite{bohigas1984characterization}. When a system has a mixed phase space, which is the common behavior of smooth Hamiltonian systems, it is found that a Berry-Robink distribution, a convex hall of the Poisson and the GOE distributions, describes the level spacing \cite{berry1984semiclassical,prosen1994numerical}. This distribution reflects the existence of eigenfunctions supported on the  islands of stability  and of eigenfunctions supported on the chaotic components of the classical phase-space \cite{backer2005flooding}. 

In quantum pseudo integrbale billiards the level spacing appears to have intermediate statistics: the nearest-neighbor distribution displays repulsion at small distances and an exponential decay at large distances  \cite{bogomolny1999models}. 

The quantum ergodicity of a system describes the wavefunction properties of classically ergodic systems, namely that almost all of them are equidistributed. 
An important result in quantum ergodicity is the equidistribution of eigenfunctions: \newline
Let $\{\Psi_n\}_{n=1}^\infty$ be an orthonormal basis of eigenfunctions of the Laplacian operator on a compact domain D. Provided the billiard flow is ergodic in phase space,  there is a density-one sequence $n_{j}\in \mathbb{N}$ such that for any  $A \subset D$: \newline
$\lim_{j\to\infty} 
    \oint_A |\Psi_{nj}|^{2}(s) ds
 =\frac{area(A)}{area(D)}$ \cite{zworski2012semiclassical}. \newline
 
This result was generalized for polygonial billiards with rational angles as it is known that the motion of a particles in such billiards is dense in configuration space for almost all directions\cite{marklof2012almost}.

\section{Main results}
We investigate eigenvalues statistics and eigenfunctions properties of a class of systems that belongs to the recently discovered family  of classical pseudointegrable Hamiltonian systems with impacts. Such systems combine motion under a smooth potential field with continuous symmetries and reflections from a corresponding family of billiards that keeps the continuous symmetries only locally and not globally. For example,  trajectories of a separable Hamiltonian  \begin{equation}
 H=H_1 + H_2, \  H_i(q_i,p_i)=\dfrac{p_i^2}{2m}+V_{i}(q_i), \ i=1,2 \label{eq:modelsham}\end{equation}
 in a right-angled polygonal billiard with at least one concave corner are pseudointegrable  \cite{becker2020impact,frkaczek2021non}. 
 
 Here, we study the quantum step oscillators: we take $V_i$ to be confining potentials which are even smooth functions with a single minimum at the origin and are monotone elsewhere, and take the right angled polygon to be   \(\mathbb{R}^{2}\setminus S_{q^{wall}}\), where \begin{equation}\label{eq:stepdef}
S_{q^{wall}}=\{(q_{1},q_{2})| \:q_1 < q_1^{wall}\le 0 \text{ and }  q_2 < q_2^{wall}\le 0\}. 
\end{equation}   The trajectories are confined by the potential and reflect from the step $S_{q^{wall}}$  \cite{becker2020impact}, see Figure \ref{fig:fig}a. Since the step boundaries are parallel to the axes, the vertical and horizontal momenta are conserved at reflections, so the motion occurs along the level sets $H_i(q_i,p_i)=E_{i}, \ i=1,2$.   Passing to the action angel coordinates of the smooth separable system, provided $E_i>V_i(q_i^{wall}),\ i=1,2$, the motion on each level set is conjugated to the directed motion on the flat cross-shaped surface,  see  Figure \ref{fig:fig}b. The direction of motion on this surface is given by $\omega_2(E_2)/\omega_1(E_1)$ and the cross shaped concave corners are at  $\{\pm\theta_1^{wall}(E_1), \pm\theta_2^{wall}(E_2)\}$, where $\omega_i(E_i)$ denotes the frequency of the smooth periodic motion under $H_i$ and $\theta_i^{wall}(E_i)$ denotes the angle of an impacting trajectory (with the convention that $\theta_i=0$ at the maximum of $q_i$). So, the direction of motion and the surface dimensions  depend continuously on $(E_1,E_2)$.
For the case of harmonic oscillators, i.e. when $V_i(q_i)=\frac{1}{2}\omega_i q_i^2$, the frequencies are fixed at $\omega_i $ and the values of $\theta_i^{wall}(E_i)$ can be explicitly computed.
Equivalently, by folding the surface, the motion on such level sets is conjugated to the directed billiard motion on an L-shaped billiard,  see Figure (\ref{fig:fig})c.  Thus,  this system is pseudointegrable \cite{becker2020impact}.
  In general, the dynamics on such surfaces has non-trivial ergodic properties. It was proven that if $q_i^{wall}<0$ for $i=1,2$, the motion is typically uniquely ergodic, and, for the case of resonant harmonic oscillators, there are level sets with co-existing periodic ribbons and dense orbits on some parts of the cross-shaped surface \cite{frkaczek2021non}.  
\begin{figure}[H]
 \centering
\begin{subfigure}{0.3\textwidth}
 
  \includegraphics[width=1.2\linewidth]{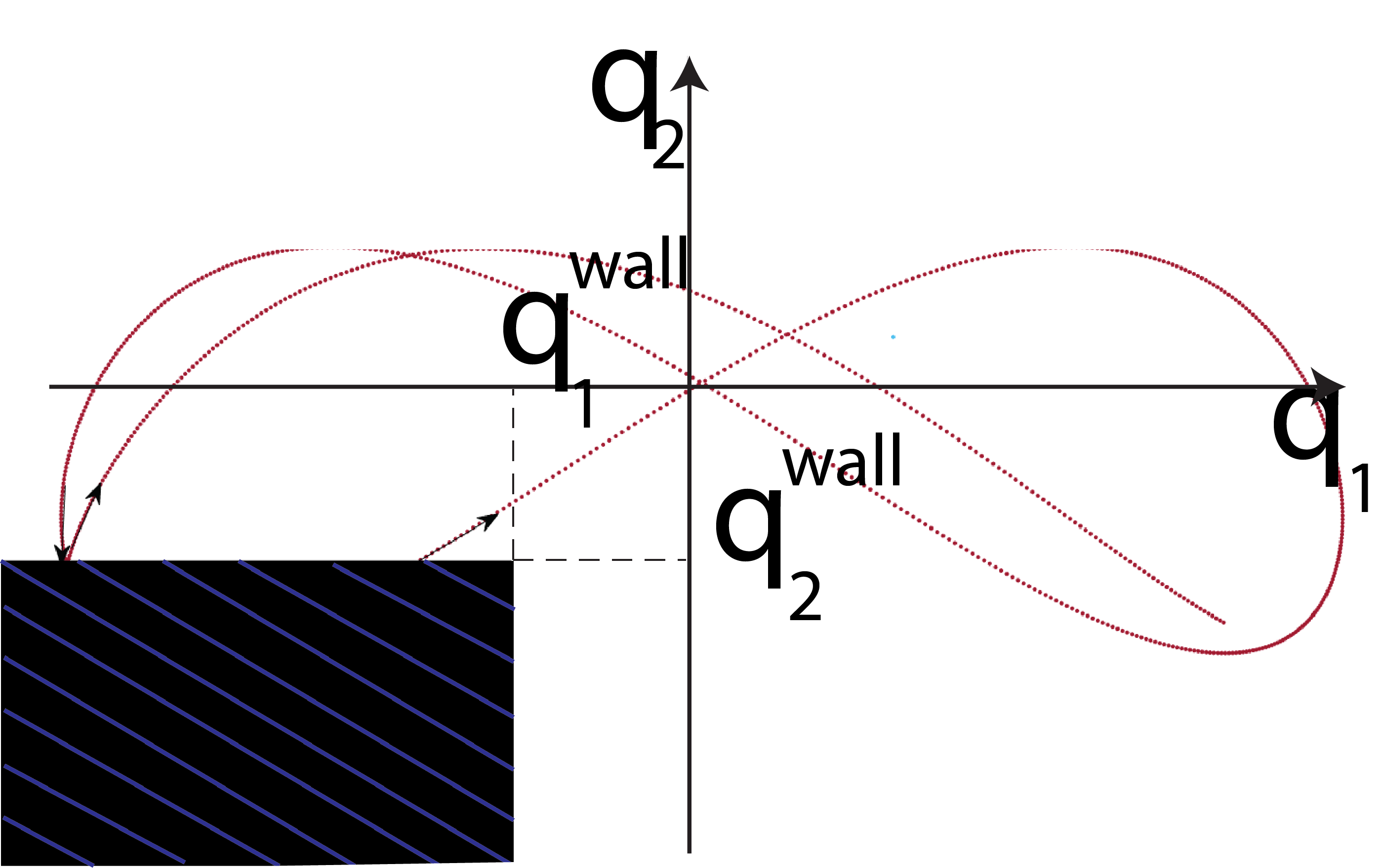}  
  \caption{}
  \label{fig:sub-first}
\end{subfigure}

\begin{subfigure}{0.2\textwidth}
 
  \includegraphics[width=1.5\linewidth]{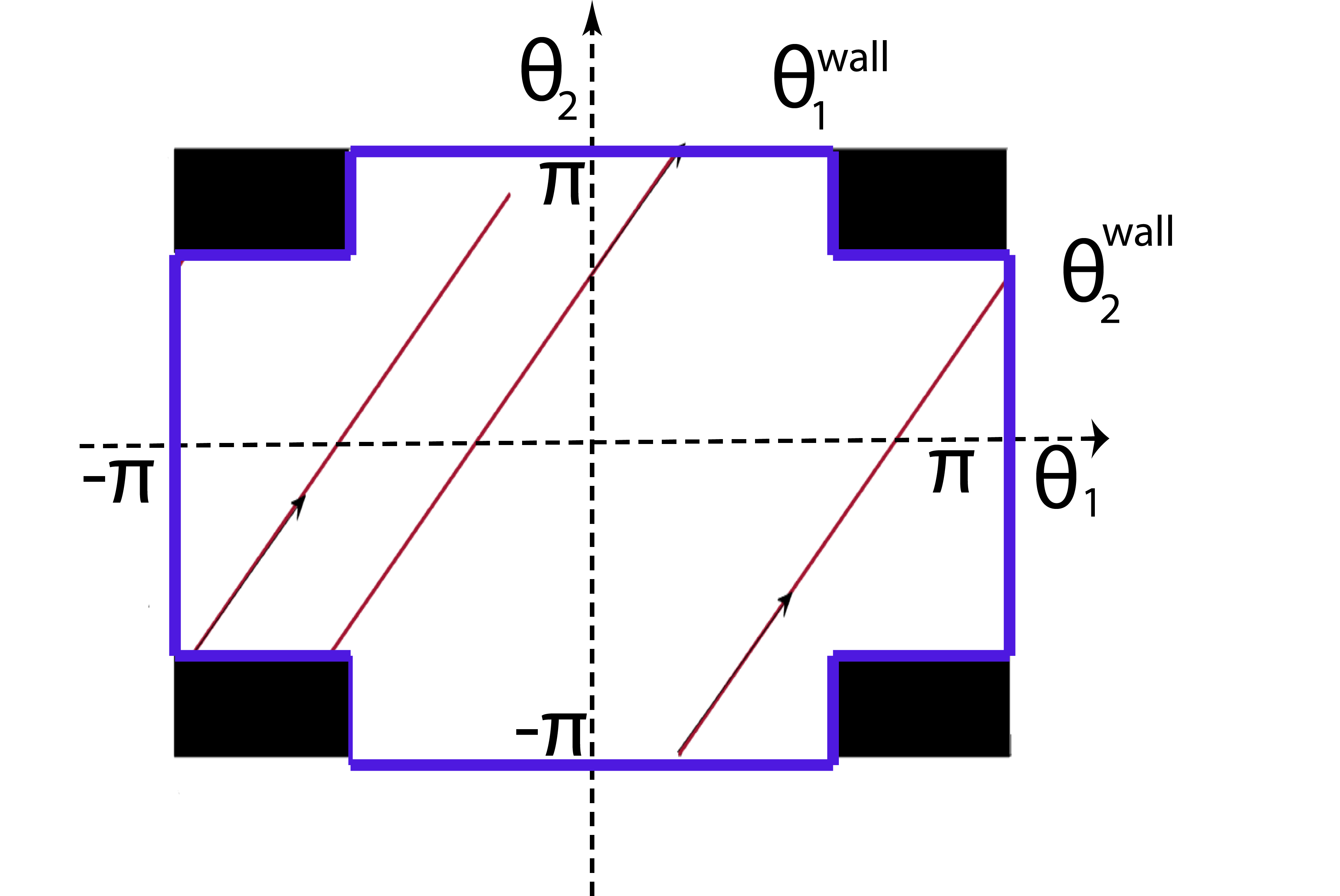}  
  \caption{ }
  \label{fig:sub-second}
\end{subfigure}
\qquad
\begin{subfigure}{0.2\textwidth}
 
  \includegraphics[width= 1.2\linewidth]{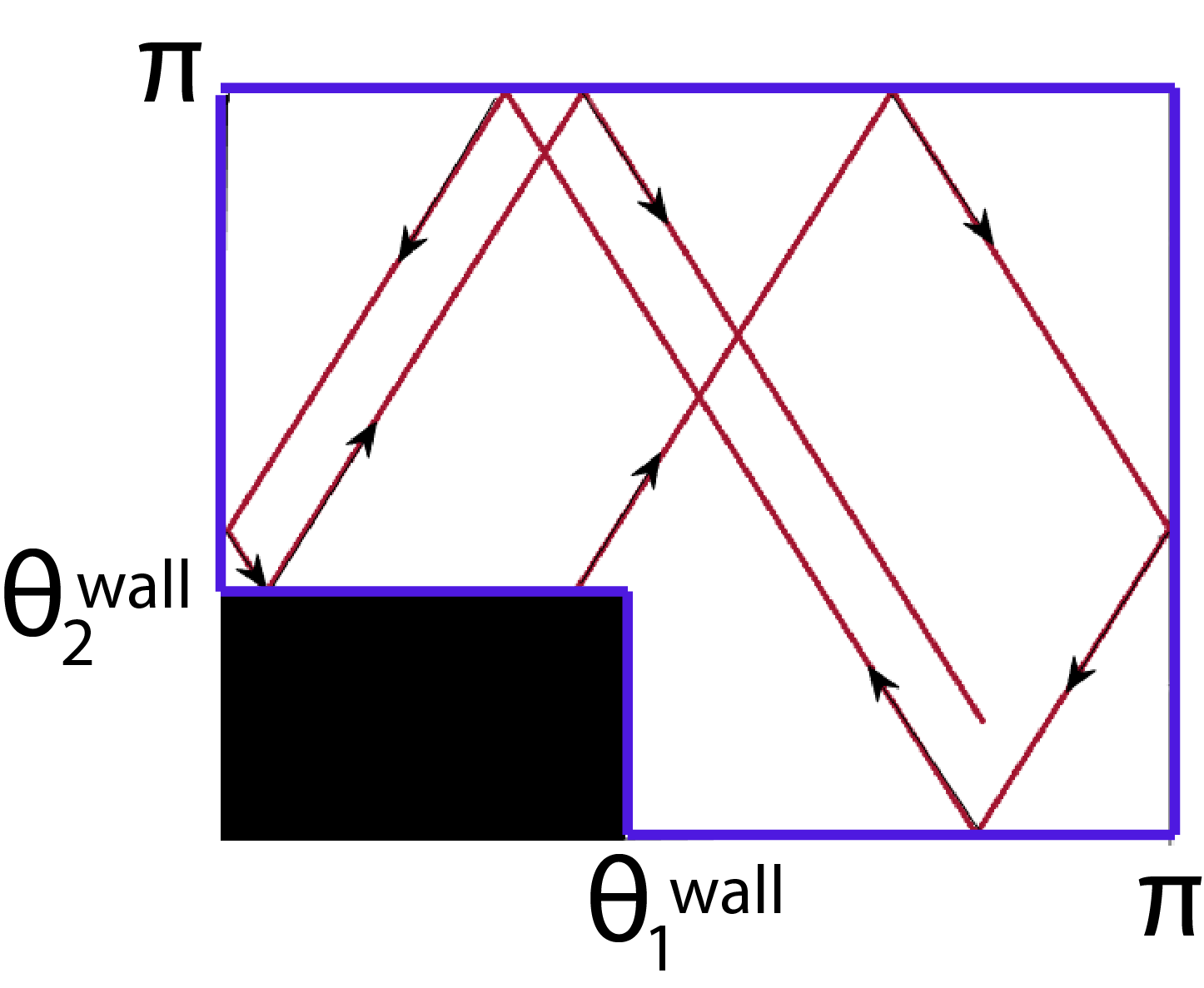}   
  \caption{}
  \label{fig:sub-second}
\end{subfigure}
\caption{A trajectory of a separable Hamiltonian  reflecting from a step.  (a)  Projection to the configuration space. (b)  The corresponding directed motion on the cross-shaped surface in the angles space. (c)  Folding the surface to the lower left quadrant leads to the corresponding billiard motion on an L-shaped billiard. Here, Eq. (\ref{eq:modelsham}) are integrated with elastic reflections from the step of Eq. (\ref{eq:stepdef}), with $V_i(q_i)=\frac{1}{2}\omega_i q_i^2, \omega_1=1,\omega_2=\sqrt{2}, q_1^{wall}=q_2^{wall}=-1, E_1=5.625,E_2=5.50 $.}
\label{fig:fig}
\end{figure}
In this work we described the quantum-classical correspondence in pseudo-integrable HIS. More specifically:

\begin{itemize}
    
    \item We describe classical and quantum dynamics of the pseudo-integrable step system for the case of a step in the origin and resonant($\omega_1=1, \omega_2=\frac{m}{n}$) harmonic potentials (all regular trajectories are periodic). In particular, we prove the existence of two families of periodic orbits co-exisisting in the same level set when n is even, and the exsitence of one family of periodic orbits when n is odd.  for m=1 we find the number of impacts and turning points along a period in each family and predict the allowed energies using EBK quantization conditions.
    
\item We describe the quantum mechanical properties(eigenfunctions structure and eigenvalues) for the more general pseudo-integrable systems (i.e system which are not completely periodic):

\begin{enumerate}
\item We find level-spacing distribution of the pseudo-integrable step system to be similar to the level-spacing distribution of the pseudo-integrable billiards(i.e semi-Poisson distribution)

\item For a step in the origin and anharmonic potential(for that case we give analytical results) we find that at least a fraction of $\frac{1}{3}$ of the eigenfunctions concentrate on the classical level sets.
\item For a step not in the origin and harmonic potential we present numerical evidence that it is also likely to have positive fraction of the eigenfunctions which concentrate on the classical level sets.
\end{enumerate}

\end{itemize}
For achieving that we used both analytic and numerical methods. The numerical method used is finite differences method via MATLAB for solving the time independent Schrodinger equation. We used the grid $[x,y]=[-15,15]X[-15,15]$ where $\delta_x=0.05$ and $\delta_y=0.05$, we set the potential to be $10^{28}$ at the step. 
\\
\\
These main results have been submitted for publication and are currently under review \cite{yaniv2022quantum}

\section{results}
\subsection{Quantization of Periodic orbits}
As we are interested in quantization, and, in particular, in studying the role of superscars(concentration of the wave functions along the classical families of periodic orbits) in the system, we look first for families of periodic orbits. Given a family of periodic orbits on a given level set $(E_1,E_2=E-E_1)$, with $\mu=(\mu_1,\mu_2)$ turning points ($\mu_1$ in the horizontal direction and $\mu_2$ in the vertical one), and $b=(b_1,b_2)$ impacts ($b_1$ with the right side of the step and $b_2$ with the upper part of the step), and an action $I(E;\mu,b)$, we can quantize it by using the EBK quantization conditions \cite{brack2018semiclassical,keller1958corrected}:
\begin{align}
\label{eqn:ebk}
I(E;\mu,b)=\hbar(n+\frac{\mu_1+\mu_2}{4}+\frac{b_1+b_2}{2}).
\end{align}
Moreover, denoting by $I_i(E_i)$ the action of the smooth $H_i$ system and by  $I_i^{wall}(E_i)=\int_{q_i\geq q_i^{wall}} p_i(q_i;E_i)dq_i=I_i \frac{2\theta_i^{wall}}{2\pi}$ the action of the impact $H_i$ system,  we obtain:
\begin{align}
\label{eqn:action}
I(E_1,E_2;\mu,b)=\sum_{i=1}^2 b_i I_i^{wall} + (\frac{\mu_i-b_i}{2})I_i,
\end{align}
namely, given  $\mu$, $b,I_i(E_i)$ and $ \theta_i^{wall}(E_i)$, we expect that the EBK quantization rule will predict the energy levels.  Yet, in general, it is non-trivial to find  $\mu$ and $b$ (see e.g. section 7 in \cite{frkaczek2021non}) nor to invert  $I(E_1,E_2;\mu,b)$ on the given family of periodic orbits.

 We consider first some simple limit cases in which periodic motion can be easily identified. When the step is at the origin ($S_0=S_{q_1^{wall}=q_2^{wall}=0}$), the corner angles are fixed at $\theta_i^{wall}(E_i)|_{q_1^{wall}=q_2^{wall}=0}=\frac{\pi}{2}$, so the dimensions of the cross-shaped surface are independent of the energy. When the potentials are harmonic, the direction of motion, $\frac{\omega_{2}}{\omega_{1}}$ is independent of the energy as well and $I_i=\frac{E_i}{\omega_i}$. Thus, by choosing resonant harmonic potentials and a step at the origin, we conclude that for all partial energies the motion is periodic and of the same type and that $I_i^{wall}=\frac{I_i}{2}$. In particular, setting: $\omega_{1}=1, \omega_{2}=\frac{n}{m}$ (with $gcd(n,m)=1$), it can be shown that there are exactly  2 options for dynamics:
 
\begin{thm}
    
For harmonic oscillator with a step in the origin with $\omega_{1}=1$ and $\omega_{2}=\frac{m}{n}$ with odd m and n,  there is one family of periodic orbits and the quantization condition is: \newline $E_{(k)}=\frac{2k}{3n}+\frac{1+\frac{m}{n}}{2}+\frac{n+m}{3n}=\frac{2k}{3n}+\frac{5(m+n)}{6n}$ \
\end{thm}

\begin{thm} 
For harmonic oscillator with a step in the origin with $\omega_{1}=1$ and $\omega_{2}=\frac{m}{n}$ where n is even and m is odd, there are two families of periodic orbits where the first family satisfies:
$\mu^{I}_{1}=2n$, $\mu^{I}_{2}=2m$ and the number of impacts $b^{I}_1$ and $b^{I}_2$ can be calculated inductively (see appendix).
The second family satisfies: $\mu^{II}_{1}=n$, $\mu^{II}_{2}=m$ and the number of impacts $b^{II}_1$ and $b^{II}_2$ can be calculated inductively(see appendix).
For the case of m=1 we can directly calculate that the number of impacts is $b^{I}_1=n$, $b^{I}_2=0$ and $b^{II}_1=0$, $b^{II}_2=1$
\end{thm}

Note that on a given resonant level set $m \omega_1(E_1) = n \omega_2(E_2) $, by rescaling time and energy: \(\hat t = \omega_1 t,\hat H =\frac{ H}{\omega_1}\), we can set $\hat \omega_{1}=1$ and $\hat \omega_{2}=\frac{m}{n}$. Clearly $b_{i}$ and $\mu_{i}$ are unchanged by the rescaling of time, so we study their dependence on \(\frac{m}{n}\) for the rescaled time case. Equivalently, the direction of the flow on the crossed surface depends only on the ratio of the frequencies. \newline
Finally by replacing the roles of $\theta_1$ and $\theta_2$ Theorem 3.2 can be applied to the case of odd n and even m. \newline

In figure \ref{figtable}, we validate the above results(Theorems 3.1 and 3.2). Notice that for even $m$ there are infinite number of energy levels at which $E^{I}_{k_1}=E^{II}_{k_2}$ (marked with green lines), and in particular, for $m=2, n=1$ and $m=2,n=5$, $E^{I}_{k_1}=E^{II}_{2k_1}$ (as shown in Fig.  \ref{figtable}).  Since the system here is symmetric, all these energy levels are degenerate, and, as shown in \ref{figtable}b and \ref{figtable}c, the common energy levels for the two families have higher degeneracy.
 
 \begin{figure}[htbp]
 \centering
 \begin{subfigure}[t]{.2\textwidth}
  \includegraphics[width=1\linewidth]{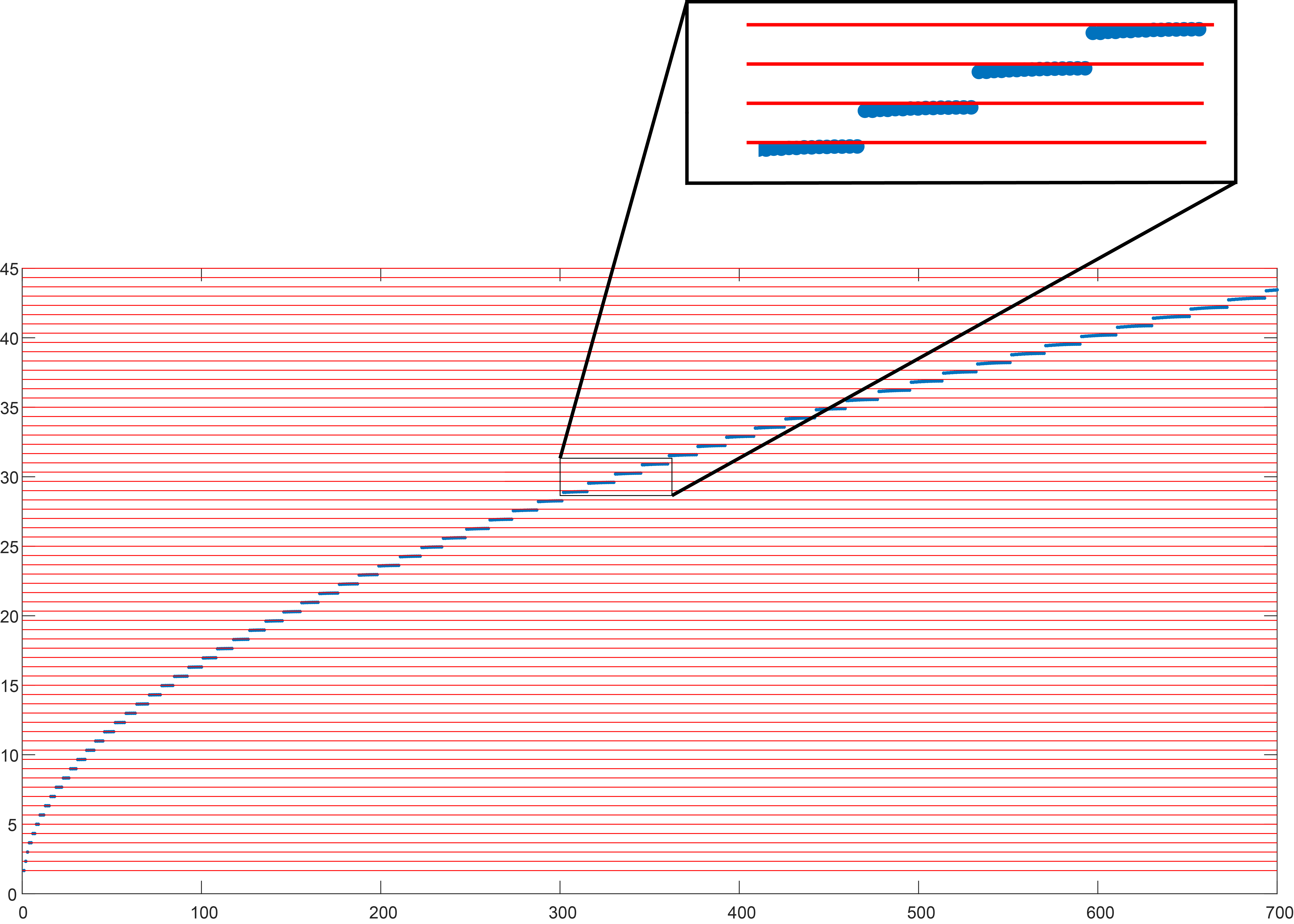}  
  \caption{}
  \label{fig:sub-second}
\end{subfigure}
\begin{subfigure}[t]{.2\textwidth}
  \includegraphics[width=1\linewidth]{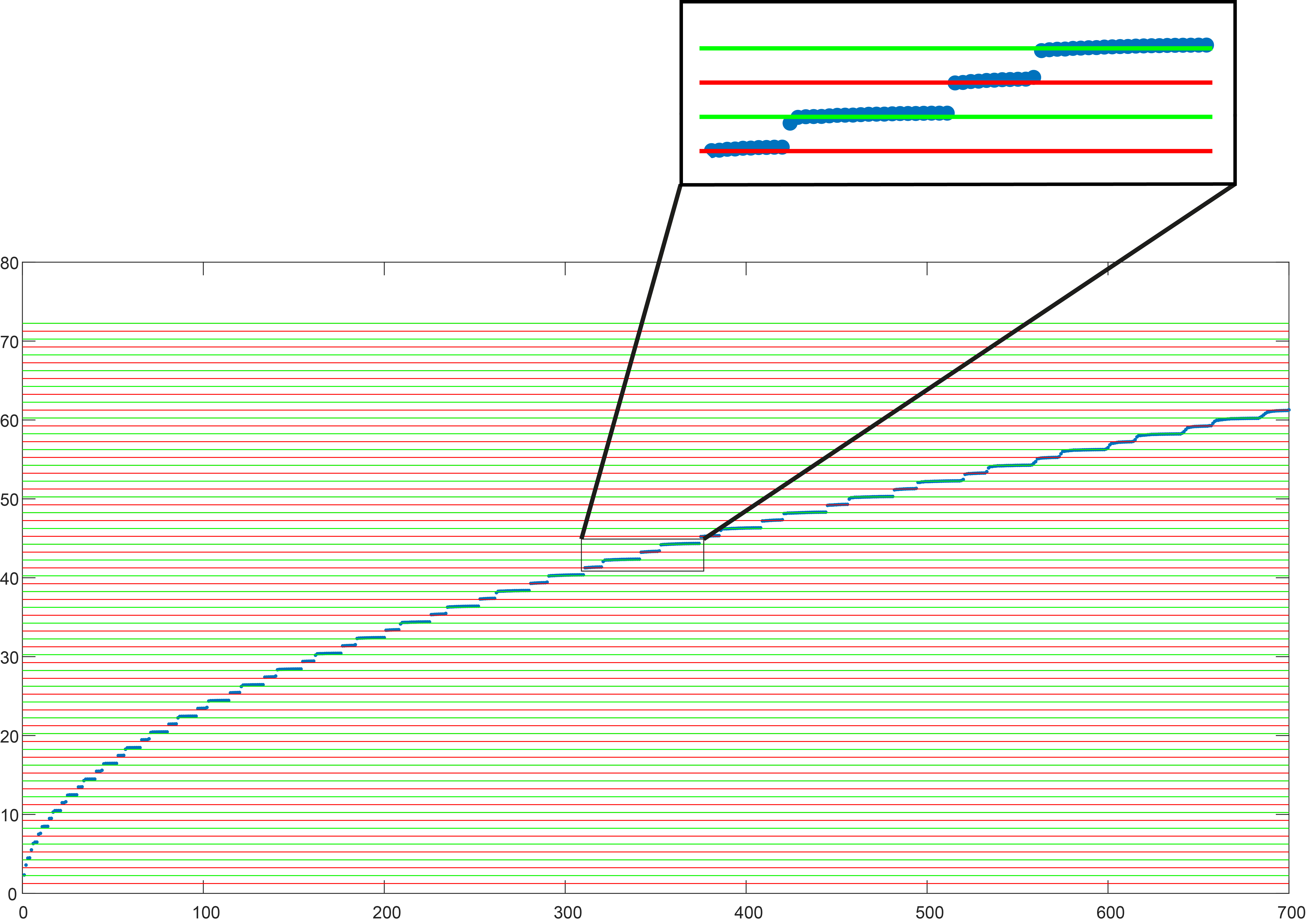}  
  \caption{}
  \label{fig:sub-second}
\end{subfigure}
\begin{subfigure}[t]{.2\textwidth}
  \includegraphics[width=1\linewidth]{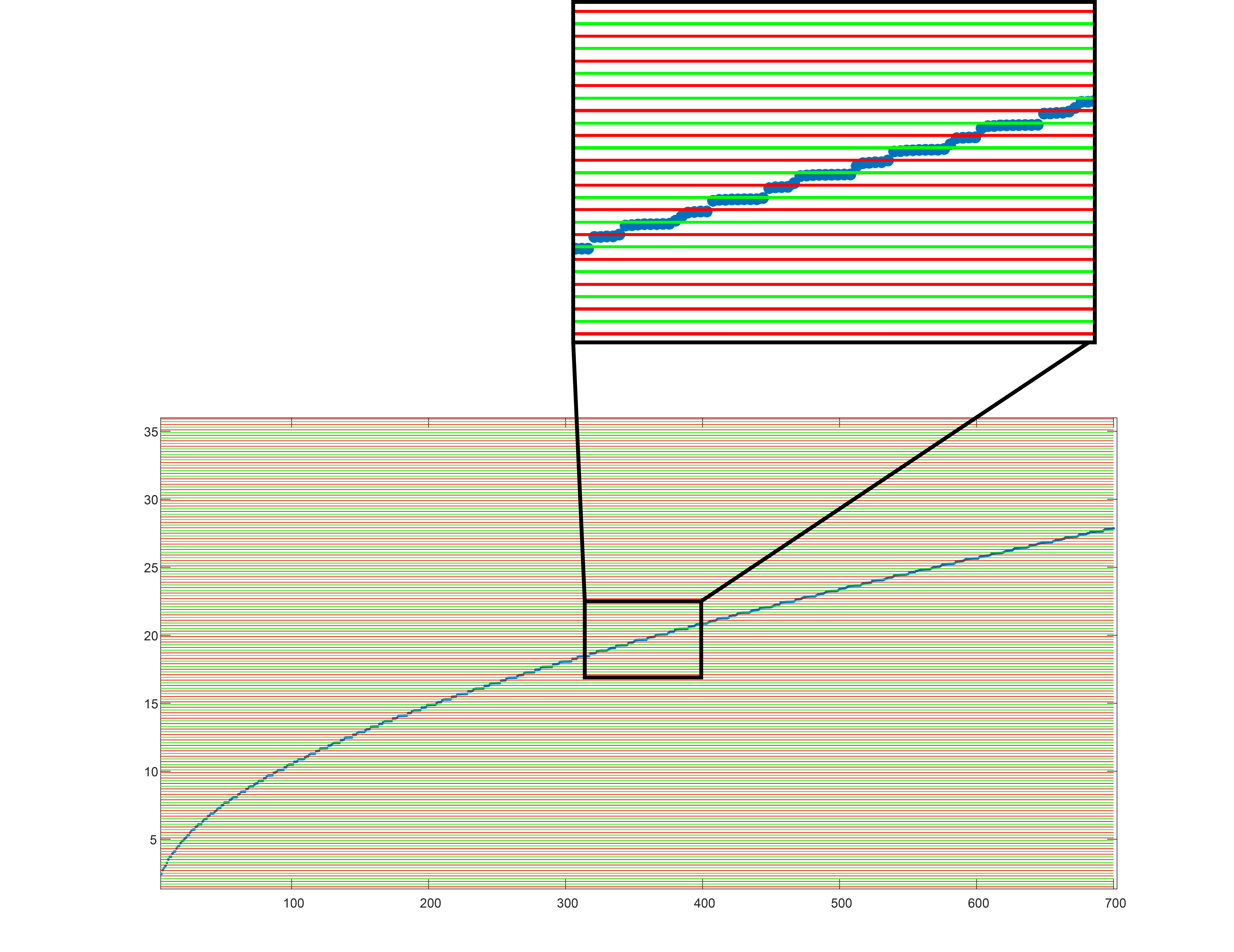}  
  \caption{}
  \label{fig:sub-second}
\end{subfigure}
    \caption{Energy levels for resonant harmonic oscillator with a step at the origin: numerical and expected (EBK) values. (a) Odd $m$ ($\omega_1=1,\ \omega_2=1$): the expected values for the single family of the periodic orbits denoted by horizontal red lines agree with the numerical values (blue dots). (b) Even $m$ $\omega_1=1,\ \omega_2=2$: The expected values (family I:red and green horizontal lines, Family II: only green horizontal lines) agree with the numerical values, and the common values have larger degeneracy.
    (c) Even $m$ and odd $n$ $\omega_1=1,\ \omega_2=\frac{2}{5}$: The expected values (family I:red and green horizontal lines, Family II: only green horizontal lines) agree with the numerical values, and the common values have larger degeneracy.}
      \label{figtable}
      \end{figure}

Next we use Weyl's law to validate our computations of correspondence between the classical families of periodic orbits and the energy levels. Recall that for the two dimensional case, Weyl's law is:
\begin{align}
 N[E_{j}: E_{j} \leq b]=
\frac{1}{\hbar^2} \text{Vol}( H\leq b)+o(1)
\text{ as } \hbar\rightarrow 0 
\end{align}
and notice that the phase space volume for the step-oscillator is:
\begin{equation}
\begin{aligned}
\text{Vol}(\mathcal{E})=\int_{0}^{I_{2}(\mathcal{E})} dI_{2}\int_0^{I_{1}(\mathcal{E}-E(I_{2})))}-4\theta_1^{wall}(I_1)\theta_2^{wall}(I_2)\\+4\pi(\theta_1^{wall}(I_1)+\theta_2^{wall}(I_2))dI_{1}.
\end{aligned}
\end{equation}
For the case of a step at the origin and harmonic oscillators, we obtain \begin{equation}
\begin{aligned}
\text{Vol}(\mathcal{E})|_{S_0,\text{Harmonic oscillators}}=\frac{3\pi^2}{2\omega_{1}\omega_{2}}\mathcal{E}^2.
\end{aligned}
\label{resvol}
\end{equation}  Fig. \ref{fig2} shows this expected correspondence.
 For the even $m$ case the contribution of the larger degeneracy associated with the energy levels which are common to the 2 different families is evident.
\begin{figure}[htbp].
  \centering
  \includegraphics[width=.9\linewidth]{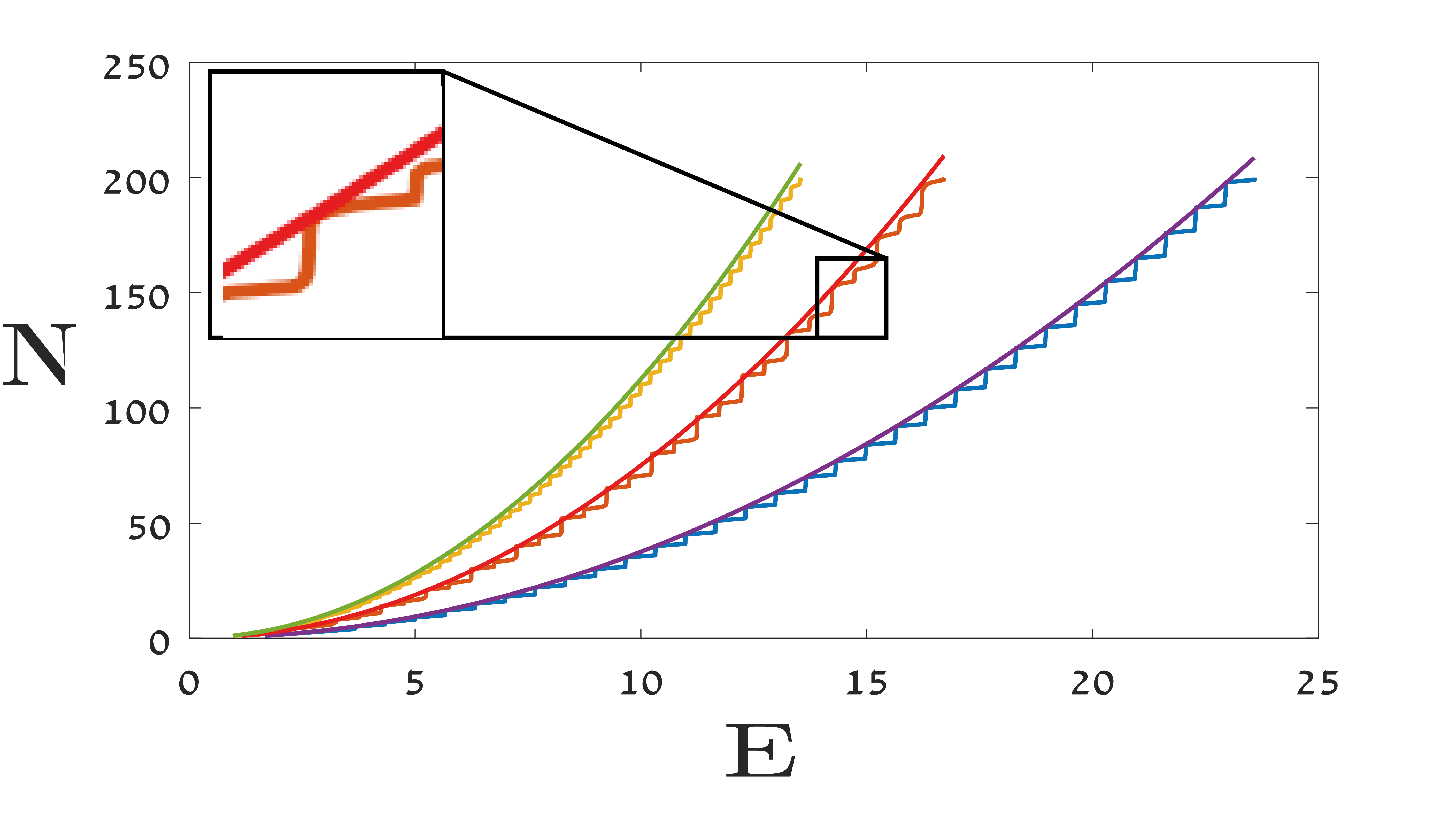}  
  \caption{Weyl's law. Smooth curves correspond to the predicted phase space volume (Eq.\ref{resvol}) for the three resonant cases ($\omega_1 =1,\omega_2=1,2,3$ yellow, red, blue lines respectively). These prediction fit the corresponding numerical results. The inset shows the non-uniform jump in $N$ for the even $m$ case.}
\label{fig2}
\end{figure}

\subsection{Wavefunctions structure and eigenvalues statistics for quantum step oscillators}

We examine first non-resonant oscillators (and not necessarily harmonic) while keeping the step at the origin.  Classically, the motion is ergodic within the level set for almost all partial energies.
Hence, we expect wavefunctions to concentrate on the projection of such level sets to the configuration space. We show that at least for a sequence of density $\frac{1}{3}$ of the wavefunctions this property holds and doesn't vanish at high energies. \newline 
In the correspondent smooth system the potential,$V=V_1(q_1)+V_2(q_2)$  is separable. Thus, its wavefunctions, $\Psi^{sm}_n$, can be written as a product of the wavefunctions of $H_i$: $\{\Psi^{sm}_n\}_{n=1}^\infty=\Psi_{1,k_1}(q_1)\Psi_{2,k_2}(q_2)$ where  $\{\Psi_{i,k_i}\}_{k_i=1}^\infty$ are the wavefunctions of the smooth one dimensional Hamiltonian $H_i$ and $E^{sm}_{n(k_1,k_2)}=E_{k_1}+E_{k_2}$. \newline Since $V_{i}$ are even:
\begin{align}
\Psi_{i,k_i}(q_{i})= \begin{cases} \Psi_{i,k_i}(-q_{i}) \phantom{-} \text{if $k_i$ is even} \\ -\Psi_{i,k_i}(-q_{i}) \phantom{-} \text{if $k_i$ is odd}
\end{cases}
\end{align}
When both $k_1$ and $k_2$ are odd, the series of wavefunctions $\{\Psi^{sm}_{n_j(k_1,k_2)}\}_{n_j=1}^\infty$ vanishes on both axes, hence, the non-smooth Hamiltonian for the case of step at the origin has a subsequence of wavefunctions of the form: 
\begin{align}
    \Psi^{S_0}_{ n_j(k_1,k_2)}(q_1,q_2)=
    \begin{cases}
    \begin{split}
    \Psi^{sm}_{\tilde n_j(k_1,k_2)}=&\Psi_{1,k_1}(q_1)\Psi_{2,k_2}(q_2),\\ &(q_1,q_2)\in \mathbb{R}^{2}/S_0      \end{split} \\
    \\ \label{product}
    \phantom{-} \phantom{-} 0 \phantom{-} & (q_1,q_2)\in S_0 
    \end{cases}
\end{align} 
These solutions are smooth in the domain ($\mathbb{R}^{2}/S_0$) and satisfy Dirichlet boundary conditions on $S_0$. Moreover, $\Psi^{S_0}_{ n_j(k_1,k_2)}$ concentrates on the projection of classical level sets; as the one-dimensional wavefunctions are well approximated by the WKB approximation \cite{brack2018semiclassical}, they decay exponentially outside of the classical allowed region of motion:
\begin{align}
\label{wkb}
\displaystyle \Psi_{i,k_i} (q_i)\approx C_{0}{\frac {e^{\theta +i \hbar ^{-1}\int {\sqrt {2\left(E_{i,k_i}-V_i(q_i)\right)}}\,dq_i}}{\hbar ^{-1/2}{\sqrt[{4}]{2\left(E_{i,k_i}-V_i(q_i))\right)}}}}.
\end{align}
Next we show that the fraction of such odd wavefunctions for the case of a step at the origin is $1/3$.
From equation \ref{eqn:ebk} for the smooth case (i.e. $b=0$) we deduce that   wavefunction that are odd in both directions (odd $k_1,k_2$) constitute one quarter of all wavefunctions:
\begin{align}
 \lim_{E\to\infty} \frac{\#\{\Psi^{sm}_{\tilde n_j(k_1,k_2)}: E_{\tilde n_j}=E^{1}_{k_1}+ E^{2}_{k_2} \leq E \}}{\#\{\Psi^{sm}_n:E_{n} \leq E\}}=\frac{1}{4}.
\end{align}
Since the step is at the origin:
\begin{align}
\text{Vol}(E)^{S_0}=\frac{3}{4}\text{Vol}(E)^{sm}
\end{align}
and thus, by Weyl's law
\begin{align}
\lim_{E\to\infty} \frac{\#\{\Psi^{S_0}_{n_j}: E_{n_j} \leq E \}}{\#\{\Psi^{S_0}_n:E_{n} \leq E\}} =
\lim_{E\to\infty} \frac{\#\{\Psi^{sm}_{\tilde n_j}: E_{\tilde n_j} \leq E \}}{\frac{3}{4}\#\{\Psi^{sm}_n:E_{n} \leq E\}} =
\frac{1}{3}.
\label{fraction}
\end{align}
We conclude that for a step at the origin there is no quantum ergodicity in configuration space, and, in fact, there is a positive measure set of eigenfunctions that concentrate on the classical level sets. 

To examine the behavior for non-symmetric pseudointegrable cases, we study numerically the shifted corner in the harmonic case: we find the level spacing of the eigenvalues and study the projections to configuration space of the eigenfunctions. Both studies propose that the shift does not break the concentration of a large subset of eigenfunctions on classical level sets.

 It is convenient for the study of the non-symmetric system to keep the step at the origin and shift the original harmonic potential to have a minimum at $(\epsilon_1,\frac{\epsilon_2}{\omega_2^{2}})=\epsilon \cdot (\cos \alpha,sin\alpha)$. Then the potential is of the form: $V=U_0+U_1$ where $U_0=\frac{q_1^{2}}{2}+\frac{\omega_2^{2} q_2^{2}}{2}$ and $U_1=-\epsilon_{1}q_1-\epsilon_{2}q_2$. Here, $\epsilon=0$ corresponds to the system with a step at the origin, and we study the behavior  for a non-resonant case at finite values of $\epsilon$, beyond the small perturbation regime. 
 Figure \ref{fig:cdf}  compares the cumulative mean level spacing distribution of this shifted potential of the first 1500 energy levels to the cumulative Poisson distribution (characterizing integrable systems, $N_p(s)=1-e^{-s}$, reflecting their locality in the classical phase space) and to the cumulative random matrix ensembles distribution, GOE (characterizing chaotic systems,  $N_W(s)=1-e^{-\frac{\pi s^2}{4}}$, reflecting their non-local nature in the classical phase space). We obtain intermediate statistics as in pseudo integrable billiards, close to semi-Poisson distribution ($N_{sp}(s)=1-e^{-2s}(2s+1)$) \cite{bogomolny1999models} (such a behaviour was also observed in a certain range of parameters in step-like time dependent one d.o.f. Hamiltonian \cite{garcia2006semi}). 
 
 Figure \ref{fig:cdf} shows that the dependence of the level spacing on $\epsilon$ appears to be mild and similar to the case $\epsilon=0$. Recall that in the case of a step at the origin, we showed that there is a positive density sequence of eigenfunctions concentrated on classical level sets. Namely, the level spacing distribution at $\epsilon=0$ reflects this locality in phase space, together with the non-locality associated with pseudointegrability.  Fig. \ref{fig:cdf} suggests that this behaviour persists  when the step is shifted from the origin.
 In fact, Fig. \ref{fig:pdf} shows that the distribution with the largest repulsion is achieved at $\epsilon=0$.   
 \begin{figure}[htbp]
 \centering
 \begin{subfigure}[t]{.4\textwidth}
  \includegraphics[width=1\linewidth]{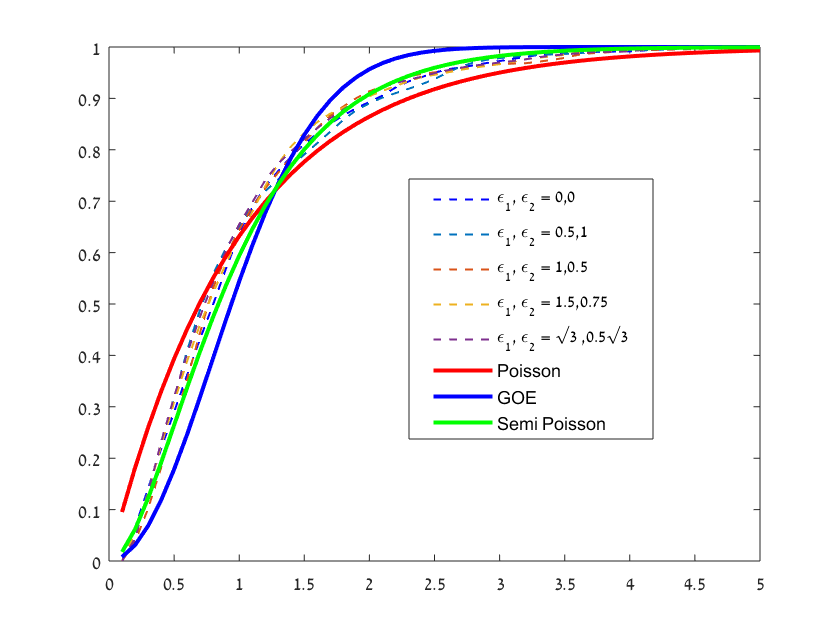}  
  \caption{} 
  \label{fig:cdf}
\end{subfigure}
\begin{subfigure}[t]{.4\textwidth}
  \includegraphics[width=1\linewidth]{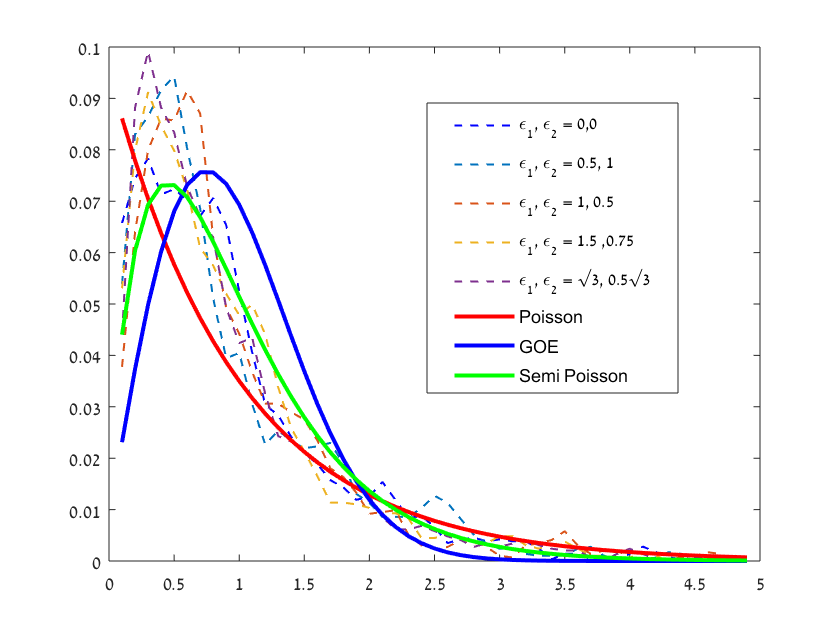} 
    \caption{}
      \label{fig:pdf}
      \end{subfigure}
      \caption{PDF and CDF of the level  spacing for a non-resonant Hamiltonian for several positions of the step. The semi-Poisson distribution (solid thick green line) provides the best fit for all positions of the step (dashed lines), including a step at the origin (blue dashed line). (a) Cumulative distribution functions of Poisson, semi-poisson GOE and numerically calculated CDFs  (b) Probability density functions of Poisson, semi-poisson GOE and numerically calculated PDFs.  The level spacing are found by a finite differences scheme for the time independent Schrodinger Eq. for the Hamiltonian \ref{eq:modelsham} with $V=U_0+U_1$ where $U_0=\frac{q_1^{2}}{2}+\frac{\omega_2^{2} q_2^{2}}{2}$ and $U_1=-\epsilon_{1}q_1-\epsilon_{2}q_2$. The step is located at the origin and is numerically represented as $V=10^{28}$. Here,  $\omega_1=1,\omega_2=\sqrt{2}$ and  $(\epsilon_1, \epsilon_2)=(0,0),(0.5,0.25),(1,0.5),(1.5,0.75),(\sqrt{3},\frac{\sqrt{3}}{2})$.}
      \label{disp}
      \end{figure}
      
To substantiate the claim that, as suggested by the level spacing plots, at large energies, the general step system still has a positive fraction of wavefunctions that concentrate on classical level sets, we calculate the wavefunctions for such systems. Since the wavefunctions depend continuously on $\epsilon$, for any given maximal  energy, for small enough $\epsilon$, such a fraction of concentrated wavefunctions exists. Hence, we first find the natural scaling of $\epsilon$ with $E$ and establish that our wavefunction calculations are far from the trivial limit of ${\epsilon \to 0}$, namely, that the perturbed wavefunctions do not correlate well with unperturbed wavefunctions.  \newline Expanding the wavefunctions in $\epsilon$,  the first order correction to $|n(\epsilon)\rangle=|n^{(0)}\rangle+ \epsilon| n^{(1)}\rangle+ O(\epsilon^2)$, is: \newline
${\displaystyle \epsilon |n^{(1)}\rangle =\sum _{k\neq n}{\frac {\langle k^{(0)}|U_1|n^{(0)}\rangle }{E_{n}^{(0)}-E_{k}^{(0)}}}|k^{(0)}\rangle } \newline
$where$ \phantom{-} U_1=-\epsilon_1 q_1-\epsilon_2 q_2$. So for large energies, the number and power of terms that contribute significantly to the sum are expected to stabilize provided we use the scaling:
 $\epsilon_1\propto\frac{E_{n+1}-E_{n}}{q_1}$ and $\epsilon_2\propto\frac{E_{n+1}-E_{n}}{q_2}$. 
Since,  for harmonic oscillators, $q_i \propto\sqrt{E}$ and  $N(E)\propto \text{Vol}(E)\propto E^{2}$, so $E_{n+1}-E_{n}\propto\frac{1}{E}$, 
we conclude that the stabilization is achieved  provided 
 $\epsilon\propto\frac{1}{E^{1.5}}$. As higher orders of the perturbation series give the same result, we actually expect that  
  ${\displaystyle |n(\epsilon)}\rangle-{\displaystyle |n^{(0)}\rangle} = O(\epsilon E_{n}^{1.5})$. 
To capture the distance between eigenfunctions of the non-perturbed Hamiltonian to the perturbed one around an energy level $E_N$, we calculate $P$, the mean squared maximal projection on unperturbed wavefunctions, and $T$, the mean number of above-threshold contributing unperturbed wavefunctions: 
\begin{equation}\label{mix}
\begin{split}
P (\epsilon,N;\Delta N,J) &=\frac{1}{\Delta N}\sum_{n=N}^{N+\Delta N} \max_{j^{0}\le J}{{|\langle  \displaystyle j^{0} \displaystyle |n(\epsilon)}\rangle}|^{2}\\
T(\epsilon,N;\Delta N,J,\delta)&=\frac{\sum_{n=N}^{N+\Delta N}  \# ({{|\langle  \displaystyle j^{0} \displaystyle |n(\epsilon)}\rangle|}^{2}>\delta)}{\sum_{n=N}^{N+\Delta N} \sum_{j^{0}=0}^J {\langle  \displaystyle j^{0}\displaystyle |n(\epsilon)}\rangle^{2}}.
\end{split}    
\end{equation}
 \newline
 Figure \ref{mixing} shows that 
 $P (\epsilon E_N^{3/2},N;\Delta N,J)$ and $T (\epsilon E_N^{3/2},N;\Delta N,J,\delta)$ are, to a good approximation,  independent of $N$,
 supporting the validity of our scaling. 
Moreover, while for small  $\epsilon (\frac{E_N}{E_{301}})^{3/2}$ we see that, as expected, there is a strong correlation between the perturbed and unperturbed wavefunctions,   for 
$\epsilon E_N^{3/2} \ge  E_{301}^{3/2} $ the maximal projection, $P$, is small while the level of mixing, $T$, is large, indicating that for such values of $\epsilon E^{3/2}$ we are indeed far from the small $\epsilon$ limit.   Additional computations show that a further increase in $\epsilon E_N^{3/2}$ leads to further decrease in $P$. 
  \begin{figure}[htbp]
 \centering
 \begin{subfigure}[t]{.4\textwidth}
  \includegraphics[width=1\linewidth]{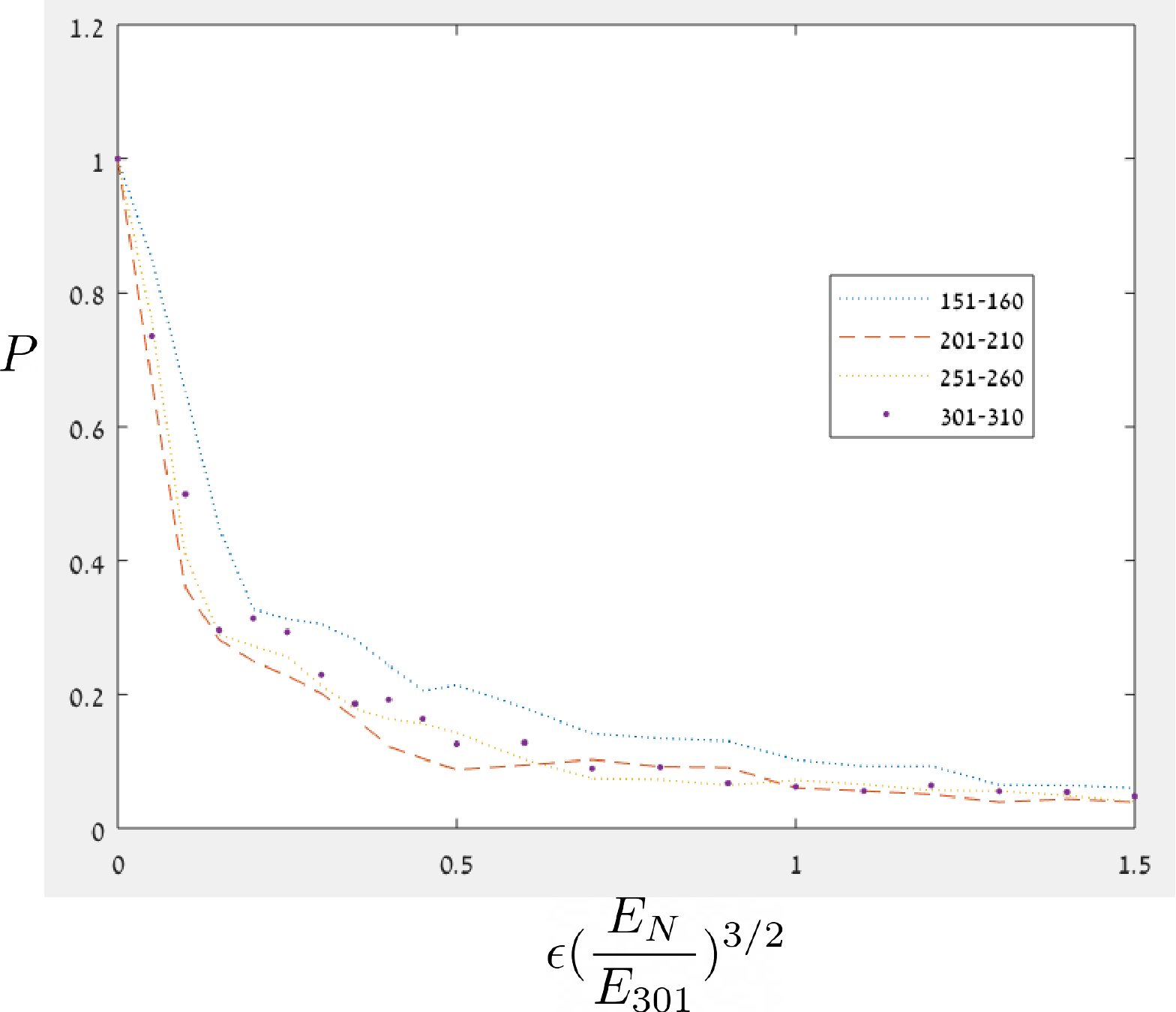}  
  \caption{}
  \label{maincomponent}
\end{subfigure}
\begin{subfigure}[t]{.4\textwidth}
  \includegraphics[width=1\linewidth]{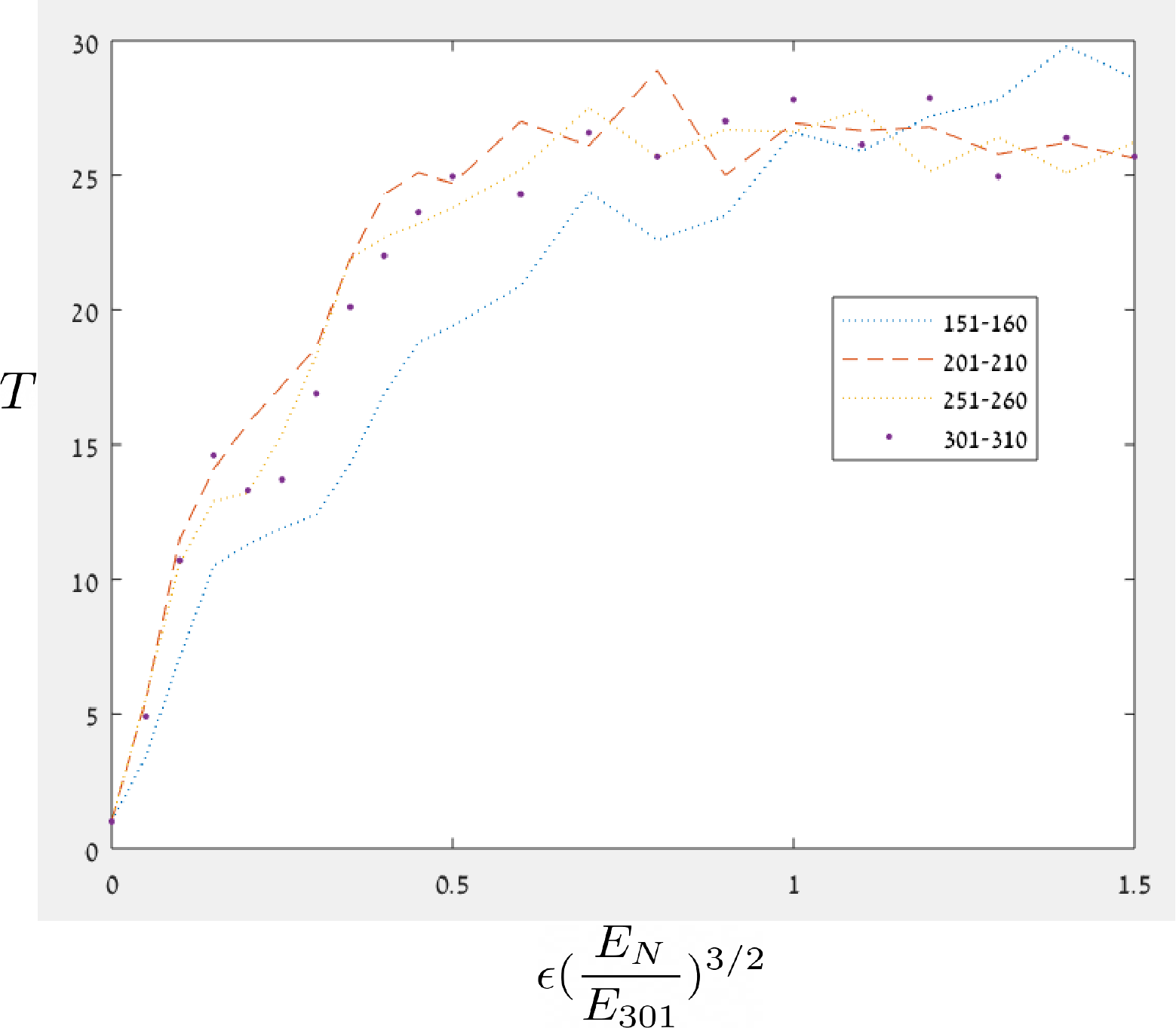}
  \caption{}
  \label{mixingrate}
\end{subfigure}
    \caption{Scaling of the perturbed wavefunctions with $\epsilon$ and energy. (a) The  mean maximal projection on unperturbed wavefunctions along the energy-scaled $\epsilon$, $\epsilon (\frac{E_N}{E_{301}})^{3/2}$:  $P (\epsilon (\frac{E_N}{E_{301}})^{3/2},N;10,400)$ (b) The mean number of above-threshold contributing unperturbed wavefunctions along the energy-scaled $\epsilon$, $\epsilon (\frac{E_N}{E_{301}})^{3/2}$:  $T (\epsilon (\frac{E_N}{E_{301}})^{3/2},N;10,400,0.01)$. These functions are plotted  
    for $N=151,201,251,301$ and for several $\epsilon=(\epsilon_1,\epsilon_2=\frac{\epsilon_1}{2})$ values.    }
      \label{mixing}
      \end{figure}
      
Finally, we show that even when   $\epsilon (\frac{E_{n}}{E_{301}})^{3/2}\gg 1$, i.e. when the wavefunctions are not well approximated by the unperturbed wavefunctions,
a substantial fraction of the  wavefunctions concentrate on classical level sets.
Figure \ref{wavefunctions} shows  the 1481-1500 wavefunctions in Logarithmic scale normalized by the maximal absolute value of the wavefunctions for the unperturbed (step at the origin) and perturbed ($\epsilon=(1.5,0.75)$) wavefuncations (so $\epsilon (\frac{E_{1500}}{E_{301}})^{3/2}=5.25$). For both the perturbed and unperturbed systems,  wavefunctions that are concentrated along the  classical level sets, i.e., are essentially restricted to the configuration space region $(q_1,q_2) \in [q_1^{min}(E_1,\epsilon_1),q_1^{max}(E_1,\epsilon_1)]\times [q_2^{min}(E_2,\epsilon_2,\omega_2),q_2^{max}(E_2,\epsilon_2,\omega_2)]
\setminus S_{q^{wall}}$ where $q_i^{max,min}$ correspond to the classical level set boundaries,  are clearly seen (e.g. see  wavefunction 1 in the unperturbed system and wavefunction 19 in the pertubed system). We call such wavefunctions concentrated wavefunctions.
 \begin{figure}
 \centering
 \begin{subfigure}[t]{.4\textwidth}
  \includegraphics[width=1\linewidth]{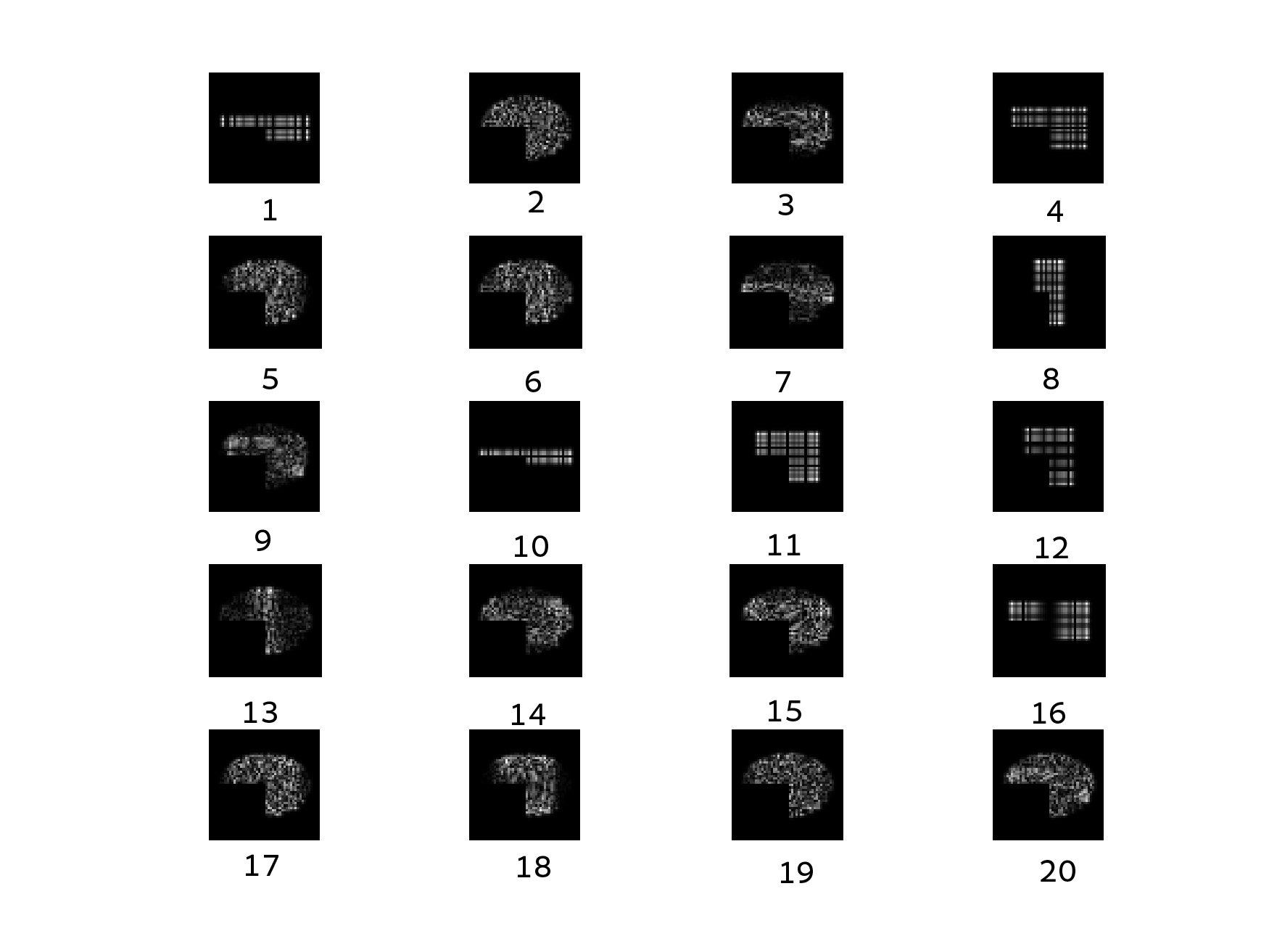}  
  \caption{ }
  \label{mediumpert}
\end{subfigure}
\qquad
\begin{subfigure}[t]{.4\textwidth}
  \includegraphics[width=1\linewidth]{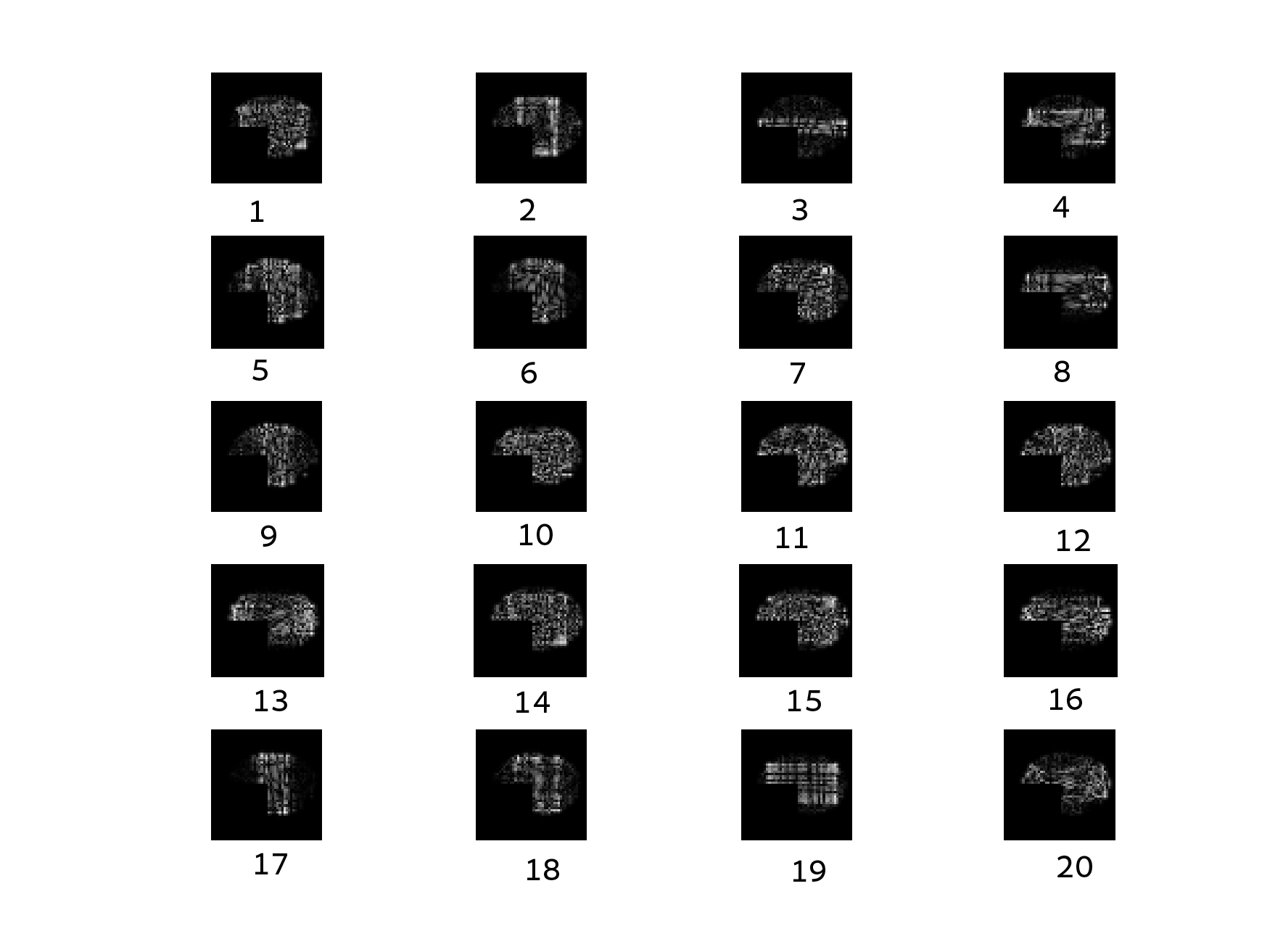}  
  \caption{ }
  \label{bigpert}
\end{subfigure}
    \caption{High energy wavefunctions for a step at the origin and for a shifted step.  (a) The unperturbed Hamiltonian.   (b) The perturb Hamiltonian  with $(\epsilon_1,\epsilon_2)=(1.5,0.75)$. The wavefunctions  for $n=1481-1500$ are plotted. To better visualize the main mass concentration we plot   
    $Log(|\Psi_n(q_1,q_2;\epsilon)|+\max_{q_1,q_2}|\Psi_n((q_1,q_2;\epsilon)|)$. }
      \label{wavefunctions}
\end{figure}

To quantify this observation, we need to distinguish between concentrated wavefunctions from wavefunctions which are not concentrated. To this aim we define vertical and horizontal means of the wavefunctions: 
\begin{equation}
\begin{split}
M^H_n(q_2)&=\int_{-\infty}^{\infty}|\Psi_n(q_1,q_2)|^{2}dq_1 \\
M^V_n(q_1)&=\int_{-\infty}^{\infty}|\Psi_n(q_1,q_2)|^{2}dq_2.     
\end{split}    
\label{means}
\end{equation}
and suggest that
\begin{equation}
\label{normalizedenergy}
\Tilde{E}=\frac{V_1(arg \max_{q_1}M^V_n(q_1))+V_2(arg \max_{q_2}M^H_n(q_2)}{E}
\end{equation}
provides a good indicator for the wavefunctions concentration: it is close to one for concentrated wavefunctions and has a much lower value for the rest of the wavefunctions.

Figures \ref{etilde}(a,b)  present  $\Tilde{E}$ values in the case of corner at the origin for low (a) and high (b) ranges of energies. Red points represent $\Tilde{E}$ values for the product wavefunctions of Eq. (\ref{product}) and constitute around $1/3$ of the  20  $\Tilde{E}$ values.
We see that some of the blue points align with the red ones, while others, around 1/5 for the lower energies and 1/2 for the higher energies have a much lower value. The insets present $(M^H_n,M^V_n)$ in the positive quadrant for the three different types of wavefunctions: for a product  wavefunction (red point, wavefunction 1 in \ref{etilde}(a) ), for a concentrated wavefunction with a similar $\Tilde{E}$ value (blue point, wavefunction 13 in \ref{etilde}(a) ) and for a non-concentrated wavefunction with a low $\Tilde{E}$ value (blue point, wavefunction 9 in \ref{etilde}(a) ). 
In the first two cases we recognize an oscillatory structure within the classically allowed region, and we observe that the maximal power appears close to the edge. In contrast,  the insets corresponding to the low $\Tilde{E}$ value show a non oscillatory structure with peaks at arbitrary positions within the Hill region. \newline
Figures \ref{etilde}(c,d) present a similar computation for the case of the shifted potential,    $\epsilon=(1.5,0.75)$, for which there are no product wavefunctions, yet concentrated and not concentrated wavefunction do appear, and the indicator  $\Tilde{E}$ seems to distinguish between these two types of wavefunctions. 

The reasoning for this suggestion is as follows;
For step at the origin, for the product wavefunctions (eq. \ref{product}), $M^H_n(q_2)=|\Psi_{n,2}(q_2)|^2$ for $q_2>0$ and $M^H_n(q_2)=|\Psi_{n,2}(q_2)|^2/2$ for   $q_2<0$, so by the WKB approximation (eq.\ref{wkb}), and similarly for $M^V_n(q_1)$,  we indeed expect $\Tilde{E}=1-f(E)$ for some function $f(E)$ which tends to zero as $E$ goes to infinity (e.g., Figures \ref{etilde}(a,b) suggest that $f(E_{500})\approx0.15, E_{500}=39.9 $ and $f(E_{1500})\approx0.1, E_{1500}=70.5$). 
For non-product yet concentrated wavefunctions on some classical configuration space region defined by the partial energies $(E_1,E_2)$,
the $arg max$ of  $M^{V,H}$ cannot be larger than the corresponding $q_i^{max}$. Moreover, as classically, one of the momenta components vanishes at the edges of the classical region, the projection of the Liuoville measure to the configuration space there is expected to be larger, hence, by the correspondence principle, we expect maximal densities near the edges. Hence,  $\Tilde{E}$ provides the approximate ratio between the sum of the potential energies at the classical region corners (belonging to the boundary of the classical Hill region) to the total energy,  so we expect it to have a similar $\Tilde{E}$ values to the corresponding product wavefunctions.
 In contrast, for a wavefunction which does not concentrate on a single classical level set we do not expect the maxima in the horizontal and vertical directions to lie necessarily on the boundary of the Hill region (see insets corresponding to the lower $\Tilde{E}$ values), thus the sum of the potential energies at such an interior point leads to a lower value of $\Tilde{E}$.

In conclusion, Figures \ref{wavefunctions}  and \ref{etilde} suggest that the fraction of concentrated wavefunctions does not vanish at high energies even when the step is shifted. \newline

\begin{figure}[H]
\centering
\includegraphics[width=0.9\linewidth]{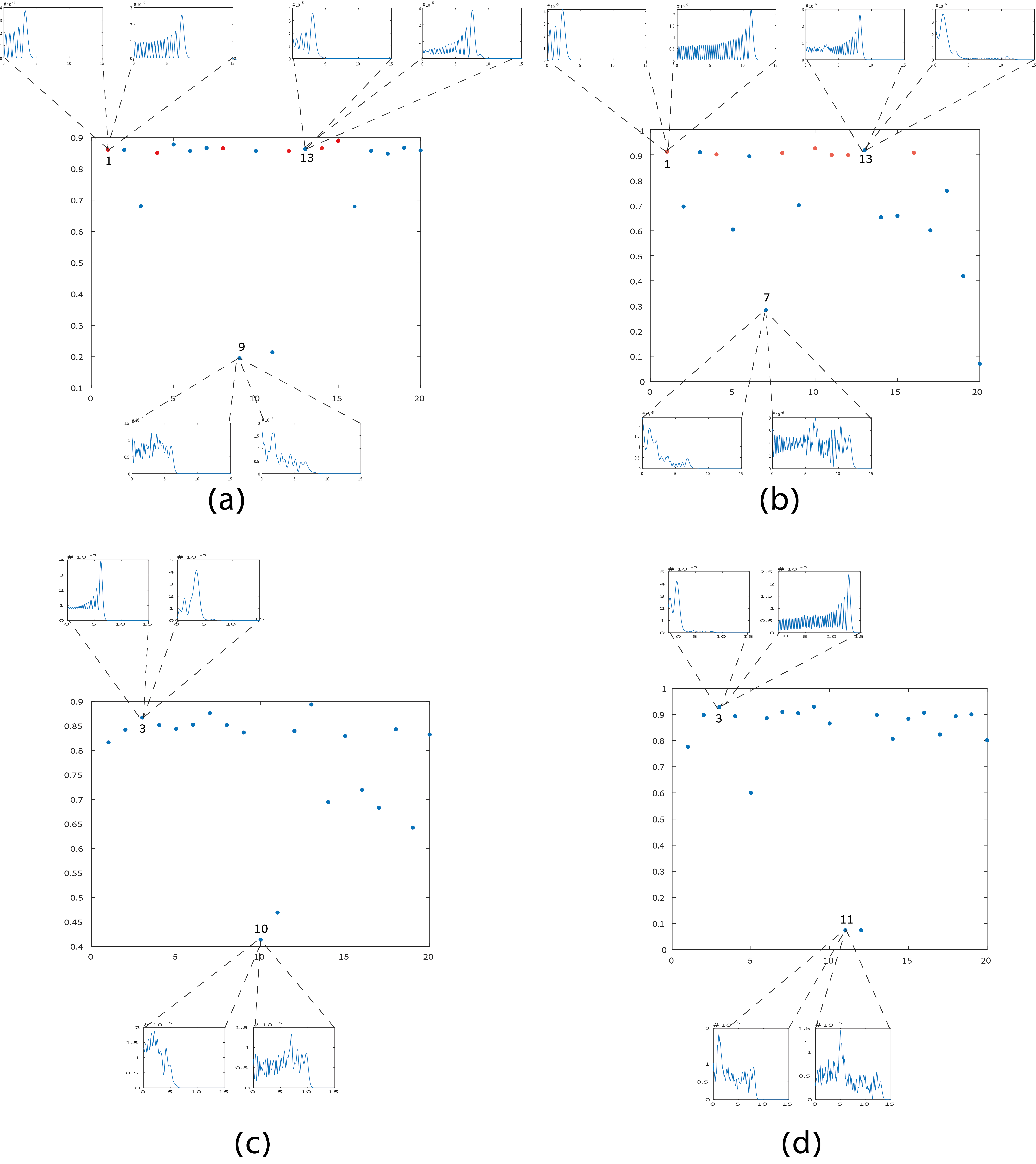}
    \caption{An indicator for the concentration of wavefunctions on classical level sets. The indicator $\Tilde{E}$ of Eq. \ref{normalizedenergy} is plotted for the case of a step at the origin, (a) 481-500 and  (b) 1481-1500 wavefunctions (the wavefunctions of Fig. \ref{wavefunctions}a). The indicator $\Tilde{E}$ is plotted for the case of a shifted step ( $\epsilon=(1.5,0.75)$)  
   (c) 481-500 and (d) 1481-1500 wavefunctions(the wavefunctions of Fig. \ref{wavefunctions}b). \newline
   The insets present $M^H,M^V$ for specific points}
      \label{etilde}
      \end{figure}
      
 \section{Summary and discussion}
 
     We studied the correspondence of a quantum step-oscillator - a two dimensional quantum oscillator in the presence of a step (a step-like region $S$ in the configuration space  at which the potential energy is infinite) to its classical analog,  a pseudointegrable Hamiltonian impact system.   For the case of harmonic resonant oscillators with a corner at the origin, for which families of periodic orbits can be explicitly constructed, we demonstrated that the EBK quantization condition provides a good predictor to the energy levels (Figure \ref{figtable}), and that Weyl's law provides a good approximation to the growth in the number of wavefunctions (Figure \ref{fig2}). Moreover, we observed that in even-resonance cases two different families of periodic orbits belonging to the same component of the level set co-exist, with distinct corresponding wavefunctions, each contributing a positive portion to the phase space volume (Figure \ref{fig2}).   This demonstrates that the non-ergodicity of level sets has a quantum analog. We showed that the intermediate level spacing of the quantum step-oscillator for non-resonant and not necessarily harmonic potential hardly depends on the position of the step (taken in the negative quadrant) and is approximately  semi-Poisson, indicating repulsion of energy levels, similar to the level spacing obtained for pseudointegrable billiards (Figure \ref{disp}).
     When the step is at the origin, we showed that there is a positive fraction of wavefunctions that remain concentrated along the classical level sets at arbitrarily high energies, as occurs for integrable systems, namely they do not tend to equidistribute  in the configuration space as is the case for pseudointegrable billiards (Eq. (\ref{product})-(\ref{fraction}) and Figures \ref{wavefunctions}a and \ref{etilde}a,b). Finally, when the corner is shifted from the origin, we conjecture, based on numerical evidence for non-resonant harmonic oscillators,  that there is a positive density series of wavefunctions which are not equidistributed and concentrated along the classical level sets (Figures \ref{wavefunctions}b and \ref{etilde}c,d).

\newpage
\section{appendix}

\begin{definition}
{A family of periodic orbits} of the step system is the family of periodic orbits having identical number of turning points and impact points $b_{i}$ and $\mu_{i}$  
\end{definition}

\begin{proposition} For a step at the origin, on level sets with $\omega_1=1, \omega_2=\frac{1}{n}$ the step system has: \newline

1) For an odd $n$: exactly one family of p.o. with $b_{1}=n,\mu_{1}=3n,b_{2}=1, \mu_{2}=3$ and action  $I=\frac{3nI_{1}+3I_{2}}{2}$. Its quantization condition is: $E_{(k)}=\frac{k}{1.5n}+\frac{5(1+n)}{6n}$  \newline
2) For an even  $n$: exactly two families of p.o. one with $\mu^I_{1}=2n,$ $b^I_{1}=n$,    $\mu^I_{2}=2$,$b^I_{2}=0$ and action  $I^I=\frac{2nI_{1}+2I_{2}}{2}$ and another one with $ \mu^{II}_{1}=n, \mu^{II}_{2}=1,b^{II}_{1}=0,b^{II}_{2}=1 $ and action  $I^{II}=\frac{nI_{1}+I_{2}}{2}$. Theirs quantization condition:  $E^{I}_{(k)}=\frac{k}{n}+\frac{4n+2}{4n}$,  $E^{II}_{(k)}=\frac{2k}{n}+\frac{n+3}{2n}$
\end{proposition}

\begin{proof}
Consider the Poincare map restricted to a given level set, $I(H_1=e_1,H_2=E-e_2)$, $P_1:\Sigma_{1}^{wall} \rightarrow \Sigma_{1}^{wall}$, so $P_1: \theta_2 \rightarrow \bar \theta_2, \theta_2 \in [0,2\pi) $. Divide this section to $2n$ equal length intervals starting from $\frac{\pi}{2}$ (so  $s_{i}=(\frac{\pi}{2}+(i-1)\frac{\pi}{n} ,\frac{\pi}{2}+i\frac{\pi}{n} \text{ mod } 2\pi), \ i=1,...,2n $). Since $\omega_1=1, \omega_2=\frac{1}{n}$ and the step is at the origin (so the width of each strip in the cross is \(\pi\)), the transition between these segments and the corfresponding increase in the turning points and impacts is (notice that each crossing of the section $\theta_i=0$ or $\theta_i=\pi$ corresponds to an additional turning point in the $q_i$ direction):  \newline
$
\begin{cases}
  s_i \rightarrow s_{i+1}, (b_1,\mu_1,b_2,\mu_2) \rightarrow (b_{1}+1,\mu_1+1,b_2,\mu_2) & \text{if: $  i <n-1 $ , $ i \neq \left \lfloor{\frac{n}{2}}\right \rfloor$ or $ i=2n$}\\ s_i \rightarrow s_{i+1}, (b_1,\mu_1,b_2,\mu_2) \rightarrow (b_{1}+1,\mu_1+1,b_2,\mu_2+1)   & \text{if: $i=\left \lfloor{\frac{n}{2}}\right \rfloor$}\\ s_i \rightarrow s_{i+2},  (b_1,\mu_1,b_2,\mu_2) \rightarrow (b_{1},\mu_1+2,b_2,\mu_2)& \text{if:$n\leq i <2n-1
   ,i \neq \left \lfloor{\frac{3n}{2}}\right \rfloor, \left  \lfloor{\frac{3n}{2}}\right \rfloor-1$}\\  s_i \rightarrow  s_{i+2},  (b_1,\mu_1,b_2,\mu_2) \rightarrow (b_{1},\mu_1+2,b_2,\mu_2+1) & \text{if: $i = \left \lfloor{\frac{3n}{2}}\right \rfloor,i = \left \lfloor{\frac{3n}{2}}\right \rfloor-1$}\\ s_i \rightarrow s_{n+1},  (b_1,\mu_1,b_2,\mu_2) \rightarrow (b_{1},\mu_1+2,b_2+1,\mu_2) & \text {if: $i=2n-1$} 
        \end{cases}
 $       
        \newline

       Thus, orbits starting at $s_{1}$ always reach $s_{n}$ after $n$ iterations, undergoing up to this point $n-1$ impacts with the right wall, $n$ turning points in the $q_1$ direction ($n$ crossing of the section $\theta_1=0$) and $1$ turning point in the $q_2$ direction. \newline 
       
       For an even $n$, $s_n$ is mapped to $s_{2n}$ after $\frac{n}{2}$ iterations without passing through $s_{2n-1}$, and then maps back to $s_1$. Namely, such orbits complete a period by visiting all segments $s_1,..,s_n$ and only even segments between $s_{n+1}$ and   $s_{2n}$. 
       Thus,  $\mu^I_{1}=2n,$ $b^I_{1}=n$,     $\mu^I_{2}=2$, $b^I_{2}=0$. \newline
       
       The other family of orbits starts at $s_{n+1}$ and visits only odd segments between $s_{n+1}$ and $s_{2n}$. Thus, it undergoes a single impact with the upper wall, has $n$ turning points in the $q_1$ direction  and a single turning point in the $q_2$ direction: $ \mu^{II}_{1}=n,  
 \mu^{II}_{2}=1, 
 b^{II}_{1}=0, b^{II}_{2}=1 $. \newline
       
       For an odd $n$, $s_n$ is mapped to $s_{2n-1}$ after $\lfloor{\frac{n}{2}}\rfloor$ iterations, then impacts the upper wall and maps to $s_{n+1}$. Then it visits only even segments between $s_{n+1}$ and $s_{2n}$ and maps back to $s_1$.
       Thus,  $b_{1}=n,\mu_{1}=3n,b_{2}=1, \mu_{2}=3$
\end{proof} 

 \begin{figure}[H]
 \centering
  \includegraphics[width=0.5\linewidth]{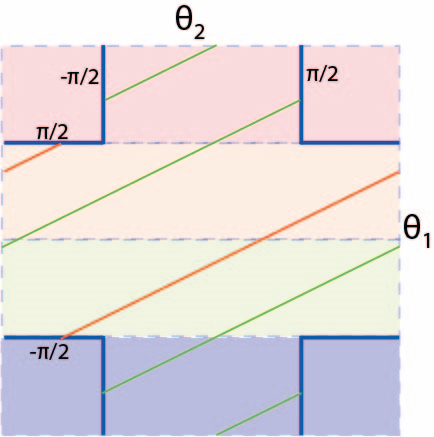}  
  \caption{Trajectories of family I(green) and family II(red) in action angle coordinates for $\omega=1$, $\omega=\frac{1}{2}$ with the c-series $[0;2]$. Proposition 5.2 states: $(\mu^{I},b^{I})=((4,2),(2,0))$ and $(\mu^{II},b^{II})=((2,1),(0,1))$.  The four different background color stands for the 4 $s_{i}$ segments.}

      \label{figtable}
      \end{figure}

Now consider $\omega_1=1,\omega_2=\frac{m}{n}$:

\begin{definition} {c-series:} Given a rational number $\frac{m}{n}=\frac{m_{0}}{n_{0}}$
the associated even continued fraction series, $[c_{0};c_{1},...,c_{k}]$ and its partial expansions, $\frac{m_{i}}{n_{i}}, i=0,..,k$  are defined inductively on $i$; Given  $\frac{m_{i}}{n_{i}}$ and  $[c_{0};c_{1},...,c_{i-1}]$ (so at $i=0$ the previous c-series is empty) define: \begin{itemize}
    \item 
If $\frac{m_{i}}{n_{i}}$ is an integer then
$c_{i}=\frac{m_{i}}{n_{i}}$ and the series ends, so \(k=i\).
\item Otherwise,
$$c_{i}= 
       \left\lfloor{\frac{m_i}{n_i}}\right \rfloor+N$$
       where $$N=\begin{cases}
       1, & \text{if:} \left \lfloor{\frac{m_i}{n_i}}\right \rfloor \text{is odd } \\
       0, & \text{if:} \left \lfloor{\frac{m_i}{n_i}}\right \rfloor \text{is even} 
       \end{cases}
       $$
       and
         $$\frac{m_{i+1}}{n_{i+1}} =
           (\frac{m_{i}}{n_{i}}-c_{i})^{-1}
         $$
         where $n_{i+1}$ always gets the positive sign.
       \end{itemize}
\end{definition}

\begin{lemma}A rational number $\frac{m_{i}}{n_{i}}$ associated with a c-series $[c_{i};...,c_{k}]$ is equal to $c_{i}+\frac{n_{i+1}}{m_{i+1}}$ where $\frac{m_{i+1}}{n_{i+1}}$ is associated with $[c_{i+1};...,c_{k}]$
\end{lemma}

\begin{proof} 
Derives directly from the definition.
\end{proof}

\begin{lemma}
The c-series is finite and unique, $k=0$ iff $\frac{m}{n}$ is an integer, and, for all $k>0$ and  $0 \le i \le k-1$,  $c_{i}$ is even, whereas  $c_{k}$ can be either even or odd. For all $i$, $c_i$ can be negative,  $c_k\ne 1$, and for all $i>0,\ c_i\ne0$.
\end{lemma}

\begin{proof}
We use proof by induction: If $k=0$,  $\frac{m}{n}$ is an integer number and the c-series ,$[c_{k=0}=\frac{m}{n}]$, is finite and $c_{k}$ can be even or odd.  \newline
Otherwise, assume  - at the $i'th>0$ step, given a non integer rational number  $\frac{m_{i}}{n_{i}}$, and a sequence $[c_{i};c_{i+1},...,c_{k}]$, we compute, 
       $\frac{m_{i+1}}{n_{i+1}}=
       \begin{cases}
       -(\frac{n_{i}-(m_{i} \mod n_{i})}{n_{i}})^{-1}, & \text{if:} \left \lfloor{\frac{m_{i}}{n_{i}}}\right \rfloor \text{is odd } \\
        (\frac{m_{i} \mod n_{i}}{n_{i}})^{-1}, & \text{ if:} \left \lfloor{\frac{m_{i}}{n_{i}}}\right \rfloor \text{is even}
       \end{cases}$ \newline
       where $n_{i+1}$ always gets the positive sign. \newline
       Thus, $|m_{i+1}|=n_{i}$ \newline
       And, $n_{i+1}<n_{i}$ \newline
       As a result of that after a finite number of steps $t\leq n_0$, $n_{i+t}=1$, $\frac{m_{i+t}}{n_{i+t}}$ is an integer and the series is finite. \newline
       In addition:
             $|\frac{m_{i+1}}{n_{i+1}}|>1$, which implies $\ c_{i}\ne0$ \newline
      It follows that for any given $\frac{m_0}{n_0}$ the process defines all $m_{i}$, $n_{i}$ and $c_{i}$ uniquely.
      \end{proof}

\begin{lemma}
given $\frac{m_{i}}{n_{i}}$ associated with $[c_{i};...,c_{k}]$ and $\frac{m_{i+1}}{n_{i+1}}$ associated with $[c_{i+1};...,c_{k}]$, the eveness properties of $(m_{i+1},n_{i+1})$ are identical to those of $(n_i, m_i)$, namely, $m_{i+1} \mod 2 =  n_{i} \mod 2$ and  $n_{i+1} \mod 2 = m_{i} \mod 2$.
\end{lemma}

\begin{proof}
by definition $$\frac{m_{i+1}}{n_{i+1}} =
           (\frac{m_{i}}{n_{i}}-c_{i})^{-1}
         $$
         where $c_{i}$ is even: \newline
       Thus,  $|n_{i+1}| \mod 2=|m_{i}-c_{i}n_{i}|\mod 2 = m_{i} \mod 2$ \newline
         $|m_{i+1}| \mod 2 =|n_{i}| \mod 2 = n_{i} \mod 2$ 
\end{proof}

\begin{lemma}
given $\frac{m_{0}}{n_{0}}$ associated with $[c_{0};...,c_{k}]$\newline
1) if $m_{0}$ is even and $n_{0}$ is odd then $c_{k}$ is even and k is even \newline
2) if $m_{0}$ is odd and $n_{0}$ is even then $c_{k}$ is even and k is odd \newline
3) if $m_{0}$ is odd and $n_{0}$ is odd then $c_{k}$ is odd \newline
\end{lemma}

\begin{proof}
$c_{k}=\frac{m_{k}}{n_{k}}$ \newline
1)
For an even $m_{0}$ and an odd $n_{0}$ we assume by contradiction an odd k: by  lemma 5.6 $m_{k}$ is odd and $n_{k}$ is even. However, $c_{k}=\frac{m_{k}}{n_{k}}$ is an integer, which creates a contradiction. Thus, the only possible option is an even k, an even $m_{k}$ and an odd $n_k(=1)$ which implies an even $c_k$.

2) For an odd $m_{0}$ and an even $n_{0}$ we assume by contradiction an even k: by  lemma 5.6 
$m_{k}$ is odd and $n_{k}$ is even. However, $c_{k}=\frac{m_{k}}{n_{k}}$ is an integer, which creates a contradiction. Thus, the only possible option is an odd k, an even $m_{k}$ and an odd $n_k(=1)$ which implies an even $c_k$.\newline
3)By lemma 5.6 for an odd $m_{0}$  and an odd $n_{0}$, $m_{k}$ is odd and $n_{k}$ is odd. Thus $c_{k}$ is odd as well.
\end{proof}

Examples:  $\frac{1}{3} = [0,3]_c$
           $\frac{3}{5}= [0,2,-3]_c$
            $\frac{2}{5}= [0,2,2]_c$
            $\frac{29}{12}=[2,2,2,2]$ $\frac{9}{8}=[2,-2,2,-2,2,-2,2,-2]_c$ 
            
\begin{lemma}: The c-series is identical to the continued fraction series iff the continued fraction is of the form
$[a_{0};a_{1}...,a_{k}]$ where $a_{i}$ is even for $0<i<k-1$ and $a_{k}$ can be either even or odd.
\end{lemma}
\begin{proof} derives directly from the definition.
\end{proof}

\begin{definition} {The right impact interval},
$J_{r}^{\omega_{2}}$,(For a level set with $\omega_{1}=1$, $\omega_{2}=\omega_{2}$) is the set of all $\theta_{2}$ at $\Sigma_{1}^{wall}$,  which are mapped by the flow to the step region, namely, satisfying $|\theta_{2}+2\theta_{1}^{wall}\frac{\omega_2}{\omega_1} \mod 2\pi| \geq \theta_{2}^{wall}$
\end{definition}

\begin{lemma}
The right impact interval of a step system with corner in the origin, with frequencies $\omega_2$ that differ by an even integer are identical: if $\omega_{1}=1$ and $\omega_{2}=\frac{m}{n}$ and  $\omega'_{1}=1$ and $\omega'_{2}=\frac{m}{n}+2L$ for an integer $L$ then $J_{r}^{\frac{m}{n}+2L}=J_{r}^{\frac{m}{n}}$.
\end{lemma}

\begin{proof}

 $J_{r}^{\frac{m}{n}}=\{ \theta_2: |\theta_{2}+\pi\frac{m}{n} \mod 2\pi|>\theta_{2}^{wall}\}$ \newline
Now since,  $|\theta_{2}+\pi\frac{m}{n}+2\pi  L \mod 2\pi|=|\theta_{2}+\pi\frac{m}{n} \mod 2\pi|$, $J_{r}^{\frac{m}{n}}=J_{r}^{\frac{m}{n}+2L}$
\end{proof}

\begin{lemma}  A step system with corner in the origin with $\omega_{1}=1$ and $\omega_{2}=\frac{m}{n}$ and a step system with corner in the origin with $\omega'_{1}=1$ and $\omega'_{2}=\frac{m}{n}+2L$ for an integer L has the same number of families of p.o
\end{lemma}

\begin{proof} 
Given $\omega_{1}=1$  and the rational frequencies $\omega_{2}=\frac{m}{n}$ and $\omega'_{2}=\frac{m}{n}+2L$ the Poincare map $P_{1}^{\omega_2}:\Sigma_{1}^{wall} \rightarrow \Sigma_{1}^{wall}$ satisfies the relation $P_{1}^{\frac{m}{n}+2L}$=$P_{1}^{\frac{m}{n}}$.This derives from lemma 5.10,  $J_{r}^{\frac{n}{m}}=J_{r}^{\frac{n}{m}+2L}$, realizing that for $\theta_1 \in [-\frac{\pi}{2},\frac{\pi}{2}]$(the vertical center of the cross) the $\theta_2$ coordinate is increased by $2\pi L$ whereas for $\theta_1 \in [\frac{\pi}{2},\frac{3\pi}{2}]$(the horizontal part of the cross) the $\theta_2$ coordinate is increased by $4\pi L$, so altogether we see that: \newline
 $P_{1}^{\frac{m}{n}+2L}= \begin{cases}
       P_{1}^{\frac{m}{n}}+2\pi  L    , & \text{if: $\theta_{2} \in J_{r}^{\frac{n}{m}}$} \\
       P_{1}^{\frac{m}{n}}+6\pi  L   , & \text{else} 
       \end{cases}
       $
       \newline 
$\mod 2\pi$, the two maps are identical and have the same dynamics
\end{proof}

\begin{lemma} A step system with corner in the origin with $\omega_{1}=1$ and $\omega_{2}=\frac{m}{n}$ and a step system with corner in the origin with $\omega'_{1}=1$ and $\omega'_{2}=\frac{n}{m}$ has the same number of families of p.o
\end{lemma}

\begin{proof}  
Given $\omega_{1}=1$  and the rational frequencies $\omega_{2}=\frac{m}{n}$ and $\omega'_{2}=\frac{n}{m}$ the Poincare map $P_{i}:\Sigma_{i}^{wall} \rightarrow \Sigma_{i}^{wall}$  satisfies the relation $P_{1}^{\frac{m}{n}}$=$P_{2}^{\frac{n}{m}}$. This derives strictly from the definition of $P_{i}$. The number of families of p.o is the number of different cycles at $P_{i}$ which has to be identical at both cases because of the above relation.
\end{proof}

\begin{lemma} A step system with corner in the origin with $\omega_{1}=1$ and $\omega_{2}=\frac{m_{i}}{n_{i}}$ associated with a c-series $[c_{i};...,c_{k}]$ and a step system with corner in the origin with $\omega'_{1}=1$ and $\omega'_{2}=\frac{m_{i+1}}{n_{i+1}}$ associated with a c-series $[c_{i+1};...,c_{k}]$ have the same number of families of p.o.
\end{lemma}

\begin{proof}  
By lemma 5.4: $\frac{m_{i}}{n_{i}}= c_{i}+\frac{m_{i+1}}{n_{i+1}}$. Thus, the lemma derives directly from lemmas 5.11 and 5.12.
\newline
\end{proof}

\begin{proposition}
A step system with corner in the origin with $\omega_{1}=1$ and $\omega_{2}=\frac{m_{0}}{n_{0}}$ associated with a c-series $[c_{0};...,c_{k}]$ has 2 families of periodic orbits if $c_{k}$ is even and one if $c_{k}$ is odd.
\end{proposition}

\begin{proof}
By lemma 5.13 a step system with corner in the origin with $\omega_{1}=1$ and $\omega_{2}=\frac{m_{0}}{n_{0}}$ associated with a c-series $[c_{0};...,c_{k}]$ and a step system with corner at the origin with $\omega_{1}=1$ and $\omega_{2}=\frac{1}{c_{k}}$ have the same number of families of periodic orbits.
By proposition 5.14 if $c_{k}$ is even there are 2 families of periodic orbits and if $c_{k}$ is odd there is one.
\end{proof}

\begin{lemma} 
For a step system with corner in the origin with $\omega_{1}=1$ and $\omega_{2}=\frac{m_{i}}{n_{i}}$ associated with a c-series $[c_{i};...,c_{k}]$ with an odd $c_{k}$($m_{i}$ and $n_{i}$ are odd) $\mu_2=3|m_{i}|$ and $\mu_1=3n_{i}$
\end{lemma}

\begin{proof} 
We prove the claim by induction: \newline
For $i=k$: \newline We recall that by proposition 5.2, $\mu_2=3|c_{k}|,\mu_1=3$.
Now assume the lemma holds for $i=k-l$: \newline
then, $\omega_{2}=\frac{m_{i}}{n_{i}}$, thus $\mu_2=3|m_{i}|$ and $\mu_1=3n_{i}$. \newline
Now, for $i-1=k-(l+1)$: \newline $\omega_{2}'=\frac{m_{i-1}}{n_{i-1}}$ associated with $[c_{i-1};...,c_{k}]$. \newline
 From lemma 5.4, $\omega_{2}'=c_{i-1}+\frac{n_{i}}{m_{i}}$ and $c_{i-1}$ is even. \newline Thus, $P^{\omega_{2}'}_{1}=P^{\omega_{2}}_{2}$(and vise-versa) and
$\mu_{2}(\omega_{2}')=|\mu_{1}(\omega_{2})+3c_{i-1}|=3|n_{i}+c_{i-1}|=3|m_{i-1}|$ and also $\mu_{1}(\omega_2')=\mu_{2}(\omega_2)=3|m_{i}|=3n_{i-1}$.
\end{proof}

\begin{lemma} 
For a step system with corner in the origin with $\omega_{1}=1$ and $\omega_{2}=\frac{m_{i}}{n_{i}}$ associated with a c-series $[c_{i};...,c_{k}]$ with an even $c_{k}$($m_{i} \mod2 \neq n_{i} \mod 2$): $\mu^I_2=2|m_{i}|$ and $\mu^I_1=2n_{i}$ and $\mu^{II}_2=|m_{i}|$ and $\mu^{II}_1=n_{i}$ 
\end{lemma}

\begin{proof} 
We prove the claim by induction: \newline
For $k=0$: \newline
We recall that by propositon 5.2, $\mu^I_2=2|c_{k}|$ and  $\mu^{II}_2=|c_{k}|$. \newline
Now assume the lemma holds for $i=k-l$: \newline
then, $\omega_{2}=\frac{m_{i}}{n_{i}}$, Thus $\mu^I_2=2|m_{i}|$ and $\mu^I_1=2n_{i}$ and $\mu^{II}_2=|m_{i}|$ and $\mu^{II}_1=n_{i}$. \newline

Now, for $i-1=k-(l+1)$:\newline
$\omega_{2}'=\frac{m_{i-1}}{n_{i-1}}$ associated with $[c_{i-1};...,c_{k}]$. \newline
By lemma 5.4 $\omega_{2}'=c_{i-1}+\frac{n_{i}}{m_{i}}$ and $c_{i-1}$ is even. \newline Thus by lemmas 5.11 and 5.12, $P^{\omega_{2}'}_{1}=P^{\omega_{2}}_{2}$(and vise-versa) and
$\mu^I_{2}(\omega_{2}')=|\mu^I_{1}(\omega_{2})+2c_{i-1}|=2|n_{i}+c_{i-1}|=2|m_{i-1}|$ and ,similarly, $\mu^{II}_{2}=|m_{i-1}|$.
For the same reason, $\mu^{I}_{1}=2n_{i-1}$ And $\mu^{II}_{1}=n_{i-1}$
\end{proof}

\begin{proposition}
A step system with corner in the origin with $\omega_{1}=1$ and $\omega_{2}=\frac{m}{n}$ 
with an odd m and odd n: $\mu_{1}=3n,\mu_{2}=3m, b_{1}=n, b_{2}=m$ \newline
\end{proposition}

\begin{proof}
By lemma 5.15 $\mu_{1}=3n$, Thus at a period, $P_{1}$ goes through 2n iterations cycle. n of the iterations are for $\theta_{2} \in J_{r}$, thus, $b_1=n$. \newline
n of the iterations are for $\theta_{2} \notin J_{r}$. 
We can write $\omega_{2}=\frac{m}{n}=\left \lfloor{\frac{m}{n}}\right \rfloor+ m \mod n$. Thus we can divide the n iterations to $m \mod n$ iterations with  $\left \lfloor{\frac{m}{n}}\right \rfloor+1$ impacts with the upper wall and $ n-m \mod n$ iterations with $\left \lfloor{\frac{m}{n}}\right \rfloor$ impacts with the upper wall.

Thus, $b_2=n\left \lfloor{\frac{m}{n}}\right \rfloor+ m \mod n=m$.
\end{proof}

\begin{proposition}
Quantized energy levels for Harmonic oscillator with corner in the origin with $\omega_{1}=1$ and $\omega_{2}=\frac{m}{n}$ for odd m and n are: \newline
$E_{(k)}=\frac{2k}{3n}+\frac{1+\frac{m}{n}}{2}+\frac{n+m}{3n}=\frac{2k}{3n}+\frac{5(m+n)}{6n}$ \newline
\end{proposition}

\begin{proof}
Quantization conditions derived from the equation: \newline
$\frac{\mu_1I_{1}+\mu_2I_{2}}{2}=k+\frac{\mu_1+\mu_2}{4}+\frac{b_1+b_2}{2}$ \newline
For harmonic oscillator: $I_{i}=\frac{E_{i}}{\omega_{i}}$ \newline
1) for an odd m and odd n: $\frac{3nE_{1}+3m\frac{E_{2}}{\frac{m}{n}}}{2}=k+\frac{3n+3m}{4}+\frac{n+m}{2}$ \newline
$E_{(k)}=\frac{k}{1.5n}+\frac{n+\frac{m}{n}}{2n}+\frac{n+ m}{3n}=\frac{2k}{3n}+\frac{5(n+ m)}{6n}$
\end{proof}

Let $b_{s,i}$ and $\mu_{s,i}$, $s\in \{1,2\}$, denote the number of impacts/turning points for a step system with corner in the origin with $\omega_{1}=1$ and $\omega_{2}=\frac{m_{i}}{n_{i}}$ associated with a c-series $[c_{i};...,c_{k}]$.
    
\begin{proposition}
A step system with a corner in the origin with $\omega_{1}=1$ and $\omega_{2}=\frac{m}{n}$ associated with a c-series $[c_{0};...,c_{k}]$,with an even m and odd n has 2 families of periodic orbits such that: $\mu^{I}_{1}=2n,\mu^{I}_{2}=2m,\mu^{II}_{1}=n,\mu^{II}_{2}=m $.\newline
The number of impacts can be calculated backwards inductively as: 
$b^{I}_{1,i}=b^{I}_{2,i+1}$, $b^{I}_{2,i}=b^{I}_{1,i+1}+c_{i}(\frac{\mu^{I}_{1,i}-b^{I}_{1,i}}{2})$, and $b^{II}_{1,i}=b^{II}_{2,i+1}$, $b^{II}_{2,i}=b^{II}_{1,i+1}+c_{i}(\frac{\mu^{II}_{1,i}-b^{II}_{1,i}}{2})$

\end{proposition}

\begin{proof}
Recall that by lemma 5.4 $\omega_2=\frac{m_{i}}{n_{i}} =
           \frac{n_{i+1}}{m_{i+1}}+c_{i}$. 
          For $\omega_{2}'=\frac{n_{i+1}}{m_{i+1}}=\frac{1}{\omega_{2,i+1}}$ we know that:
   $b'^{I}_{1}=b^{I}_{2,i}$, $b'^{I}_{2}=b^{I}_{1,i}$, and $b'^{II}_{1}=b^{II}_{2,i}$, $b'^{II}_{2}=b^{II}_{1,i}$. 
   Now as $c_{i}$ is even the step systems with $\omega_{2}'=\frac{n_{i+1}}{m_{i+1}}$ and
   $\omega_{2}=\frac{n_{i+1}}{m_{i+1}}+c_{i}$  have the same Poincare map $P_{1}$. By definition the number of rounds the periodic trajectory makes for both $\omega_2$ and $\omega'_2$, in the horizontal part of the cross($\theta_1 \in [\frac{\pi}{2},\frac{3\pi}{2}]$)  is $\frac{\mu_{1,i}-b_{1,i}}{2}$ and in each such an iteration the $\omega_2$ system has  $c_{i}$ additional impacts with the upper wall.
   The number of impacts with the right wall($b_1$) doesn't change as $P_{1}$ does not change.
   
\end{proof}

\newpage
\bibliographystyle{abbrv}
\bibliography{main}
\newpage
\section{Acknowledgments}
I would like to express my deepest gratitude to my advisor, prof. Vered Rom-Kedar
for all the help provided, the time invested in my training and the knowledge I gained during the process. I have had a great time
\end{document}